\def\lncs{0}

\documentclass[11pt]{article}
\usepackage[colorlinks]{hyperref}
\hypersetup{linkcolor=[rgb]{0,0.4,0.7},filecolor=[rgb]{0,0.4,0.7},citecolor=[rgb]{0,0.4,0.7},urlcolor=[rgb]{0,0,0.7}}
\usepackage{xspace}
\usepackage{fullpage}
\usepackage{boxedminipage}
\usepackage[boxed]{algorithm}
\usepackage{epigraph}
\usepackage[sc]{mathpazo}
\usepackage{amsmath}
\usepackage{accents}
\usepackage{caption}
\usepackage{subcaption}
\usepackage[normalem]{ulem}
\usepackage{amsthm}
\usepackage{amssymb}
\usepackage{amsfonts}

\ifnum\lncs=0
\usepackage{amsthm}
\fi
\usepackage{amsmath,amssymb}

\ifnum\lncs=0

\newtheorem{theorem}{Theorem}
\newtheorem{definition}{Definition}

\newtheorem{claim}{Claim}

\newtheorem{corollary}{Corollary}

\fi

\usepackage{tikz}
\usepackage{relsize}
\usepackage{ctable}
\usepackage{cleveref,aliascnt}

\crefname{claim}{Claim}{Claims}

\newcommand{\A}{{\mathcal A}}
\newcommand{\B}{{\mathcal B}}
\newcommand{\C}{{\mathcal C}}
\newcommand{\cH}{{\mathcal H}}

\newcommand{\cL}{{\mathcal L}}

\newcommand{\N}{{\mathbb{N}}}

\newcommand{\F}{{\mathbb{F}}}

\newcommand{\HH}{{\mathbb{H}}}

\newcommand{\bit}{\{0,1\}}

\newcommand{\ie}  {i.e.,\ }
\newcommand{\eg}  {e.g.,\ }

\newcommand{\etal}{{et~al.\ }}
\newcommand{\etalcite}[1]{{et~al.~\cite{#1}}}
\newcommand{\poly}{\mathsf{poly}}
\newcommand{\polylog}{\mathsf{polylog}}

\newcommand{\concat}{\circ}

\newcommand{\verifier}{\mathcal{V}}
\newcommand{\prover}{\mathcal{P}}

\newcommand{\ith}[1]{{#1}\textsuperscript{th}}

\newcommand{\timec}{\tau}

\newcommand{\ignore}[1]{}

\newcommand{\class}[1]{\textsf{#1}}

\newcommand{\NP}{{\class{NP}}}

\newcommand{\AM}{{\class{AM}}}
\newcommand{\PP}{{\class{P}}}
\newcommand{\NC}{{\class{NC}}}

\newcommand{\problems}[1]{\textsc{#1}}
\newcommand{\SetEquality}{\problems{SetEquality}}
\newcommand{\Clique}{\problems{Clique}}
\newcommand{\Distinctness}{\problems{Distinctness}}
\newcommand{\GNI}{\problems{GNI}}
\newcommand{\MST}{\problems{MST}}
\newcommand{\Sym}{\problems{Sym}}
\newcommand{\DSym}{\problems{DSym}}
\newcommand{\ASym}{\problems{Asym}}
\newcommand{\Permutation}{\problems{Permutation}}
\newcommand{\LeaderElection}{\problems{LeaderElection}}

\newcommand{\AMProtocol}[1]{\textsf{#1}}
\newcommand{\dAM}{\AMProtocol{dAM}}

\newcommand{\dAMAM}{\AMProtocol{dAMAM}}
\newcommand{\dMAMAM}{\AMProtocol{dMAMAM}}
\newcommand{\dMAM}{\AMProtocol{dMAM}}

\newcommand{\dIP}{\AMProtocol{dIP}}
\newcommand{\IP}{\AMProtocol{IP}}
\newcommand{\ID}{\AMProtocol{ID}}

\newcommand{\VerifierAlg}{{\mathcal{V}}}
\newcommand{\ProverAlg}{{\mathcal{P}}}

\newcommand{\parent}{\mathsf{parent}}
\newcommand{\depth}{\mathrm{depth}}
\newenvironment{boxfig}[2]{\begin{figure}[#1]\fbox{\begin{minipage}{\linewidth}
				\vspace{0.2em}
				\makebox[0.025\linewidth]{}
				\begin{minipage}{0.95\linewidth}
					{{
							#2 }}
				\end{minipage}
				\vspace{0.2em}
	\end{minipage}}}{\end{figure}}

\newcommand{\pprotocol}[4]{
	\begin{boxfig}{h!}{
			\begin{center}
				\textbf{#1}
			\end{center}
			#4
			\vspace{0.2em} } \caption{\label{#3} #2}
	\end{boxfig}
}

\newcommand{\protocol}[4]{
	\pprotocol{#1}{#2}{#3}{#4} }

\title{The Power of Distributed Verifiers in Interactive Proofs}

\author{Moni Naor\thanks{Department of Computer Science
		and Applied Mathematics, Weizmann Institute of Science Israel. Supported in part by a
		grant from the Israel Science Foundation (no.\ 950/16). Incumbent of the Judith Kleeman Professorial Chair.}
		\and Merav Parter\thanks{Department of Computer Science
			and Applied Mathematics, Weizmann Institute of Science Israel. Supported in part by grants from the Israel Science Foundation (no.\ 2084/18)}
		\and Eylon Yogev\thanks{Department of Computer Science, Technion, Haifa, Israel. Supported by the European Union's Horizon 2020 research and innovation program under grant agreement no.\ 742754.}}

\date{}

\begin{document}
\maketitle
\thispagestyle{empty}

\begin{abstract}
We explore the power of interactive proofs with a distributed verifier. In this 
setting, the verifier consists of $n$ nodes and a graph $G$ that defines their 
communication pattern. The prover is a single entity that communicates with all 
nodes by short messages.
The goal is to verify that the graph $G$ belongs to some language in a small
number of rounds, and with small communication bound, \ie the proof size.

This interactive model was introduced by Kol, Oshman and Saxena (PODC 2018) as a generalization of non-interactive distributed proofs. They demonstrated the power of interaction in this setting by constructing protocols for problems as Graph Symmetry and Graph Non-Isomorphism -- both of which require proofs of $\Omega(n^2)$-bits without interaction.

In this work, we provide a new general framework for distributed interactive proofs that allows one to
translate standard interactive protocols (i.e., with a centralized verifier) to ones where the verifier is distributed with a proof size that depends on the computational complexity of the verification algorithm run by the centralized verifier.
We show the following:
\begin{itemize}
\item Every (centralized) computation that can be performed in time $O(n)$ can
be translated into three-round distributed interactive protocol with
$O(\log n)$ proof size. This implies that many graph problems for sparse graphs
have succinct proofs (\eg testing planarity).

\item Every (centralized) computation implemented by either a small space or by uniform $\NC$ circuit can
be translated into a
distributed protocol with $O(1)$ rounds and $O(\log n)$ bits proof size for the low space
case and $\polylog(n)$ many rounds and proof size for $\NC$.

\item We also demonstrate the power of our compilers for problems not captured
by the above families.
We show that for Graph Non-Isomorphism, one of the striking demonstrations of
the power of interaction, there is a 4-round protocol with $O(\log n)$ proof size, improving upon the $O(n \log n)$ proof size of Kol et al.

\item For many problems we show how to reduce proof
size
below the naturally seeming barrier of $\log n$. By employing our RAM compiler, we get a 5-round
protocols with proof size $O(\log \log n)$ for a family of problems including Fixed
Automorphism, Clique and Leader Election (for the later two problems we
actually get $O(1)$ proof size).
\item Finally we discuss how to make these proofs non-interactive {\em arguments} via random
oracles.
\end{itemize}
Our compilers capture many natural problems and demonstrates the difficultly in
showing lower bounds in these regimes. 
\end{abstract}

\newpage
\thispagestyle{empty}
\tableofcontents
\newpage

\setcounter{page}{1}

\section{Introduction}

\epigraph{\textsf{That rug really tied the room together.}}{\textit{
		Big Lebowski, Coen Brothers, 1991}}

Interactive proofs are an extension of non-determinism and have proven to be
a fundamental tool in complexity theory and cryptography. Their development has led
us, among others, to the exciting notions of zero knowledge proofs
\cite{GoldwasserMR89,GoldreichMW91} and probabilistically checkable proofs
($\mathsf{PCP}$s).

Interactive proof is a protocol between a randomized verifier and a powerful
but untrusted prover. The goal of the prover is
to convince the verifier regarding the validity of a statement, usually stated
as membership of an instance $x$ to a language $\mathcal{L}$. The two main
requirements of the protocol are:  {\em completeness}: a verifier should accept with high probability (or probability one if we want perfect completeness) a true statement if the prover is honest, and {\em soundness}: if the statement is false,
then for any dishonest (unbounded) prover participating in the protocol the
verifier should reject with high probability (over its internal random coins).
In the classical case, the prover is computationally all powerful and the
verifier runs in polynomial time.
In a celebrated result, interactive proofs are proved to be very powerful
allowing for efficient verification of any language in $\mathsf{PSPACE}$ with a polynomial verifier
\cite{LundFKN92,Shamir92}. Another striking result was the \cite{GoldreichMW91} protocol for Graph Non-isomorphism (GNI).

Interactive proofs are largely concerned with verifiers that are
computationally bounded, but are relevant for verifiers with any sort
of limitation (\eg finite automata~\cite{DworkS92,Condon92}).
They have been studied in other settings such as communication
complexity~\cite{BabaiFS86,GoosPW18} and their connection to circuit
complexity~\cite{KarchmerW90,AaronsonW09,Williams16} and property
testing~\cite{RothblumVW13,GurR18}. Of particular interest to us are
interactive proofs for graph problems in $\PP$ with a
presumably weaker verifier (e.g.\ in
$\NC$)~\cite{GoldwasserKR15,ReingoldRR16} and a polynomial prover (i.e., prover restricted to polynomial computation).
Our results however also capture problems that go beyond $\PP$.

One schism in interactive proofs is whether the verifier has some private coins where the prover does not get to see them
(as in the original~\cite{GoldwasserMR89}) or if all coins are public (as
in~\cite{BabaiM88}), usually denote with AM for Arthur-Merlin. Goldwasser and
Sipser~\cite{GoldwasserS89} gave a compiler for converting private-coins into
public-coins that is relevant for polynomial-time verifiers. When applied to
the protocol of \cite{GoldreichMW91} for Graph Non-isomorphism it yields a two
round
public-coins (AM) protocol for showing that two graphs are {\em not}
isomorphic.

In this work we study interactive proofs where
the verifier is a {\em distributed system}: a network of nodes that interact
with a single untrusted prover. The prover sees the
entire network graph while each node in the network has only a \emph{local}
view (i.e., sees only its immediate
neighbors) in
the graph. The goal of the prover is to convince the nodes of a \emph{global} statement
regarding the network. The two main complexity measures of the
protocol (which we aim to minimize) are the number of rounds, and the
size of the proof (i.e., communication bound between the network and the prover). In this context, we ask:
\begin{center}
{\em What is the power of interactive proofs with a distributed verifier?}
\end{center}

The notion of interactive proofs with a distributed verifier was introduced
recently by Kol, Oshman and Saxena~\cite{KolOS18} as a generalization of its
non-interactive version known as ``distributed NP'' proofs (in its various
versions, e.g., \cite{KormanKP10,GoosS16,fraigniaud2011local}). The prover interacts with the nodes of the network in rounds. In each round, a node $u$ sends the prover a random challenge $R_u$. Then, the prover responds by sending each node $u$ its respond $Y_u$. Nodes can exchange their proof $Y_u$ only with their immediate neighbors $N(u)$ in the network in order to decide whether to accept the proof.
For accepting a proof all nodes must accept and to reject it is enough that one node rejects.

A simple example for a ``distributed NP'' proof is 3-coloring of a graph: the prover
gives each node
in the graph its color, and nodes exchange colors with their neighbors to
verify the validity of the coloring. In such a case, we say that the proof size
is a constant (each color can be described using two bits). Korman et.\
al.~\cite{KormanKP10} introduced this notion as a ``proof labeling scheme'' and showed that there is a long list of
problems for which a short distributed proof exists. Other problems (see also e.g.\ \cite{GoosS16}) requires proofs with $\Omega(n^2)$ bits, and thus cannot be distributed in any non-trivial manner.

There is a long line of research on the power of distributed proofs focusing on different notions of ``proof''. For example, G\"{o}\"{o}s and  Suomela~\cite{GoosS16}
studied distributed proofs that can be verified with a constant-round verification algorithm. Baruch et al.~\cite{MorFP15} studied the power of a randomized verifier in distributed proofs and Fraigniaud et al.~\cite{FraigniaudHK12}  studied the effect of such proofs when nodes are anonymous. Feuilloley et al.~\cite{FeuilloleyFH16} considered the first interactive proof system which consists of three players: a centralized prover, disprover and a distributed network verifier.
We further discuss these works in \Cref{sec:related}.

Kol \etal \cite{KolOS18} took an important step towards understanding
the power of interaction in distributed proofs. As an analog to the class
$\mathsf{AM}$ (Arthur-Marlin), they defined the class $\dAM[f(n)]$ to contain
all $n$-vertex graph problems  that admit a two-message
protocol where the communication between the prover and each node in network
is bounded by $f(n)$. As in $\AM$, the protocols in this class must be
``public-coins'', that
is, the node's messages to the prover are simply independent random bits (no
other randomness is allowed). The class $\dMAM[f(n)]$ is defined similarly for
three-message protocols (and so forth), and in general $\dIP[r,f(n)]$ denotes
protocols with $r$ rounds and communication complexity bounded by $f(n)$.

Their main positive results are for two problems $\Sym$ and $\GNI$ which have an $\Omega(n^2)$-bit lower bound in the non-interactive setting \cite{GoosS16}. In the problem of $\Sym$, the network should decide whether the network graph has
a non-trivial automorphism. In $\GNI$ problem, the goal is to decide whether the network graph is not
isomorphic to an additional input graph. Specifically,
they show that $\Sym \in \dMAM[O(\log n)]$ and that $\GNI \in \dAMAM[O(n\log
n)]$. This is a huge improvement over the $\Omega(n^2)$ lower bound for the non interactive version of this problems. On the hardness/impossibility side they show an (unconditional) lower bound for the
$\Sym$ problem for
{\em two}-message protocols: if $\Sym \in \dAM[f(n)]$ then $f(n) \in \Omega(\log
\log n)$.\footnote{The authors of \cite{KolOS18} also reported an improvement of $\Omega(\log n)$, see~\cite{Oshman18}.}

\subsection{Our Results}\label{sec:our-results}

\textbf{The Model.} We follow the distributed interactive proof model of Kol et al.\ \cite{KolOS18}:
The protocol proceeds in rounds in which nodes exchange (short) messages with the prover as well as with their neighbors in the graph.
The messages that the nodes send to each other are essentially the proofs they received from the prover. Thus in the model of \cite{KolOS18} the nodes are assumed to get from the prover their own proof as well as the proofs of their neighbors.
Note that in the distributed interactive setting, a proof size of $O(n^2)$ bits is a trivial upper bound for all graph problems, since the nodes are computationally unbounded (i.e., only their information on the network is bounded).
Our key results demonstrate that many natural graph problems on sparse graphs admit logarithmic-size proofs. In addition, it is also possible to go below the $\log n$-regime (even for dense graphs), and obtain $O(\log\log n)$-size proofs for a wide class of problems.

\paragraph{General Compiler for RAM-Verifiers.}
One of our key contribution is in presenting general methods for converting
``standard'' interactive
proofs (\ie proofs where the verifier is a centralized algorithm) to protocols
where with distributed verifier. The cost of this transformation in terms of the \emph{proof-size} depends on the computational complexity of the \emph{centralized} verification algorithm.
Our first result concerns a RAM-verifier (i.e., where the verification algorithm runs a RAM machine).
We show a general compiler that takes any $r$-protocol with a RAM-verifier with
verification complexity $\timec$ (i.e., the time complexity of the verification
algorithm is $\timec$ operations over words on length $\log n$) and transforms
is into an $r+2$-round distributed
interactive protocol with proof-size $\timec \log n/n$.
Specifically, for a verifier that runs in time $O(n)$ a $\dIP[r+2,O(\log n)]$
distributed protocol.
\begin{theorem}\label{thm:ram-compiler}
Let $\pi \in \IP$ be an $r$-round public-coin protocol for languages of
$n$-vertex
graphs where the verifier is a RAM program with running time $\timec$, then $\pi
\in \dIP[r+2,O(\timec\log n /n)]$.
In particular, if $\pi \in \NP$ and the verifier runs in time $O(n)$
then $\pi \in \dMAM[O(\log n)]$.
\end{theorem}
The benefit of this compiler is in its generality: the transformation works for
any problem while paying only in the running-time of the verifier. This is particularly useful
when the graph is sparse. For instance, it is possible to verify whether a
graph is planar in $\dMAM[O(\log n)]$ using the linear time algorithm for
planarity \cite{HopcroftT74}. Any other linear time algorithms on sparse graphs
can be applied as well.  As we will see next, this compiler is used as a basic
building block in many of our protocols. Even for those that concern
\emph{dense} graphs, and even for those the go \emph{below} the $\log n$
regime. On of the most notable example for the usefulness of this compiler is
for the problems of graph non-isomorphism and related variants.

\paragraph{Graph Isomorphism and Asymmetry with $O(\log n)$-bit Proofs.}
We combine our linear-RAM compiler with the the well-known Goldwasser-Sipser
$\GNI$ protocol~\cite{GoldwasserS89,GoldreichMW91}. Note that the $\GNI$
problem involves two graphs, and its definition in the distributed setting
might be interpreted in two ways (either the second graph is also a
communication
graph, or not). For this reason, we start by considering an almost equivalent
problem of ``graph asymmetry'' ($\ASym$) where the prover wishes to prove that
the communication graph has {\em no} (non-trivial) automorphism. The protocol for $\GNI$ can
be naturally augmented to this problem we well.
Since the running
time of the verifier in the (centralized) protocol is linear in the size of the graph, applying our compiler immediately yields that $\ASym \in
\dAMAM[O(n\log n)]$ and $\GNI \in [O(n\log n)]$ (for any definition of $\GNI$) which matches the result of \cite{KolOS18} for the same problem.

To achieve the desired bound of $O(\log n)$ proof-size, we will not use the compiler as a black-box. Instead our strategy is based on
first reducing the problem to one that is verifiable in linear-time (in the number of vertices) using $O(\log n)$-bits of proofs. Then in the second phase, we will apply the RAM-compiler on this reduced problem, using proofs of size $O(\log n)$-bit again. Our end result is a $\dAMAM[O(\log n)]$ protocol for $\ASym$, an exponential
improvement over Kol et al.\ protocol. This also applies to the $\GNI$ problem where both graphs are part of the communication network (see in case they do not correspond to the communication graph below). In contrast, for proof labeling schemes there is an $\Omega(n^2)$ lower bound~\cite{GoosS16}.
\begin{theorem}
	$\GNI \in \dAMAM[O(\log n)]$, and $\ASym \in \dAMAM[O(\log n)]$.
\end{theorem}

One of the tools used for the compiler is a protocol for the permutation problem $\Permutation$. Here, each node has a value $\pi_i$ and we need to verify that these values form a permutation over $\{1,\ldots,n\}$. We give an $\dAM$ protocol for this problem using proofs of size $O(\log n)$. This was posed as an open problem by the authors of \cite{KolOS18}.\footnote{The  problem was posed in the Interactive Complexity workshop at the Simon's Institute~\cite{Oshman18}.}
\begin{theorem}
	$\Permutation \in \dAM[O(\log n)]$.
\end{theorem}

\paragraph{Compilers for small space and low depth verifiers.}
If we allow even more rounds of communication, then we can capture a richer class of languages. Specifically, we show how to leverage the
RAM-compiler to transform the protocols of Goldwasser, Kalai and Rothblum~\cite{GoldwasserKR15} and Reingold, Rothblum and Rothblum~\cite{ReingoldRR16}
into distributed protocols. The result is that any low space (and poly-time) computation can be
compiled to constant-rounds distributed protocols with $O(\log n)$ proof size
and any ``uniform $\NC$'' (circuits of polylog depth, polynomial size and unbounded
fan-in) computations can be compiled into a distributed protocol
with $\polylog(n)$ rounds and proof size. The main work performed by the
verifier in both of these protocols is interpreting the input as a function and
evaluating  its low degree extension at a random point. We show how to
implement this using a distributed verifier. See more details in~\Cref{sec:RRR}.
This is true also for the case when the computation verified can be performed by a low depth (uniform) circuit, but in this case we need a number of rounds proportional to the depth of the circuit. 
\begin{theorem}\label{thm:intro-small-space-depth}
	Let $L$ be a language.
\begin{enumerate}
\item There exists a constant $\delta$ such that if $L$ can be decided in time
$\poly(n)$ and space $S=n^{\delta}$ then $L \in
\dIP[O(1),O(\log n)]$.
\item If $L$ is in uniform $\NC$ then $L \in 	
\dIP[\polylog(n),\polylog(n)]$.
\end{enumerate}
\end{theorem}
This can be used in turn for the $\GNI$ problem and obtain a  $\dIP[O(1),O(\log n)]$ 
even for the case where one of the graphs does not correspond to the communication 
graph. 
Another example is verifying the a tree is a minimal spanning tree ($\MST$). One can 
verify 
that a tree is a MST by a centralized algorithm with small space. Thus, we get that $\MST 
\in \dIP[O(1), O(\log n)]$. Without interaction, there is a matching upper bound and lower 
bound of $O(\log n \log W)$, where $W$ is an upper bound on the weights 
\cite{KormanK07}.


\subsection{Below the $O(\log n)$-Regime}
At this point, there is still a gap between our above mentioned results and the $\Omega(\log \log n)$ lower bound of \cite{KolOS18}. One reason for this gap is that constructing protocols with $o(\log n)$ proofs seems quite hard. The prover is somewhat limited as the basic operations such as pointing a neighboring node, counting, specifying a specific node ID, all require $\log n$ bits.

Perhaps surprisingly, we show that using our RAM compiler with additional rounds of interaction can lead to an exponentially improvement in the proof size for a large family of graph problems.
Obtaining these improved protocols calls for developing a totally new infrastructure that can replace the basic $O(\log n)$-primitives (i.e., with a logarithmic proof size) with an equivalent $O(\log \log n)$-primitives (e.g., verifying a spanning tree).
While these do not yield a full RAM compiler, they are indeed quite general and can be easily adapted to classical graph problems. Two notable examples are $\DSym$ and problems that can be verified by computing an aggregate function of the vertices.

The $\DSym$ problem is similar to the $\Sym$ problem except that the automorphism is fixed and given to all nodes. This problem was studied by \cite{KolOS18} where they showed that $\DSym \in \dAM[O(\log n)]$ but any distributed $\NP$ proof for requires a proof of size $\Omega(n^2)$. We show that using a five message protocol we can reduce the proof size to $O(\log \log n)$:
\begin{theorem}
	$\DSym \in \dMAMAM[O(\log \log n)]$.
\end{theorem}
Depending on the problem, our techniques can be used to get even smaller
proofs. In particular, if the aggregate function is over constant size elements
then the proof be of constant size.
For example, we show that the $\Clique$ problem can be solved using a proof of
size $O(1)$ in only three rounds. In contrast, without interaction, there is an
$\Omega(\log n)$ lower bound \cite{KormanKP10}.
\begin{corollary}
	$\Clique \in \dMAM[O(1)]$.
\end{corollary}
For instance, we show an $\dMAM[O(1)]$ protocol for proving that the graph is {\em not two-colorable}.
This is in contrast to non-interactive setting \cite{GoosPW18} that requires $\Omega(\log n)$ bits for this problem.
Another interesting example is the
``leader election'' problem where it is required to verify the exactly one
nodes in the network is marked as a leader. As this problem can also be casted
as an aggregate function of constant sized element, we get:
\begin{corollary}
	$\LeaderElection \in \dMAM[O(1)]$.
\end{corollary}

\paragraph{Argument Labeling Schemes.}
Can the interaction be eliminated? The simple answer for that is {\em no}! We
have observed by now several examples where few rounds of interactions break
the non-interactive lower bounds (e.g.\ for Symmetry and Asymmetry).
However, this does not seem to be the end of the story. In the centralized setting there are various techniques for eliminating interaction from protocols, especially public-coins ones. A ``standard'' such technique is the Fiat-Shamir transformation or heuristic that converts a public-coins interaction to one without an interaction. Here, we assume that parties have access to a \emph{random oracle}, and that the prover is computationally limited: it can only perform a bounded number of queries to the random oracle.
In such a case, we end up with an ``argument'' system rather than with a ``proof"
system. In an argument system proofs of false statements {\em exist} but it is
computationally hard to find them. Therefore, such protocols do not contradict
the lower bounds for proof labeling schemes. We call such a protocol an
``argument labeling scheme''. These systems can have significant savings in
distributed verification systems. More details are in
\Cref{sec:conclusions}.


\subsection{Related Work}\label{sec:related}
The concept of distributed-NP is quite broad and contains (at least) three frameworks.
This area was first introduced by Korman-Kutten-Peleg \cite{KormanKP10} that formalized the model of
proof-labeling schemes (PLS). In their setting, communications are restricted to happen exactly once between neighbors.  A more relaxed variant is \emph{locally checkable proofs} (LCP) \cite{GoosS16} introduced by G{\"o}{\"o}s and Suomela which allows several rounds of verification in which nodes can also exchange their inputs.
The third notion which is also the weakest is \emph{non-deterministic local decision} (NLD) introduced by Fraigniaud-Korman-Peleg~\cite{fraigniaud2011local}. In NLD the prover cannot use the identities of the nodes in its proofs, that is the proofs given to the nodes are oblivious to their identity assignment.

We note that when allowing prover--verifier interaction some of the differences between these models disappear. At least in the $O(\log n)$-proof regime, using more rounds of interactions allows the nodes to send their IDs to the prover, and the prover can use these IDs in its proofs. Our protocols with $o(\log n)$-bit proofs are not based on the actual identity assignment, but rather only on their port ordering.

Prior to the distributed interactive model of \cite{KolOS18}, Feuilloley, Fraigniaud and Hirvonen \cite{FeuilloleyFH16} considered the first interactive proof system which consists of three players: a centralized prover, a decentralized disprover and a distributed verifier (the network). This model gives considerably more power to the verifier as it can get some help from the strong disprover. \cite{FeuilloleyFH16} showed that such interaction between a prover and a disprover can considerably reduce the proof size. The most dramatic effect is for the nontrivial automorphism problem which requires $\Omega(n^2)$ bits with no interaction, but can be verified with $O(\log n)$ bits with two prover--disprover rounds.

Very recently, Feuilloley et al.~\cite{FeuilloleyFHPP18} considered another generalization of \cite{KolOS18} where instead of allowing several rounds of interaction between the prover and the verifier, they allow several \emph{verification} rounds. That is, the prover gives each node a proof at the first round, it then disappears and the nodes continue to communicate for $t$ many rounds.
They showed that for several ``simple" graph families such as trees, grids, etc.\ every proof labeling with $k$ bits,
can be made an $k/t$-bit proof when allowing $t$ verification rounds.
Note that our distributed protocols can simulate such a scheme, but since our protocols use a small number of interactive rounds, the reduction in the proof size that we get from the framework of \cite{FeuilloleyFHPP18} in negligible.

\section{Our Techniques}
\subsection{The RAM Program Compiler}
Many non-interactive distributed proofs (known as ``proof 
labels'') \cite{KormanKP10} are based on the basic primitive of verifying that 
a given marked subgraph is a spanning tree \cite{AfekKY97}.
In particular, in most of these applications, the subgraph itself is given as part of the proof to the nodes
(i.e., a vertex $u$ gets its parent $\parent(u)$). I.e., the prover computes a spanning tree for the vertices to facilitate the verification of the problem in hand (e.g., cliques, leader election etc.). Indeed, throughout we will use the prover to help the network compute various computations to facilitate the verification of the problem in hand.
We start by briefly explaining the proof labeling of spanning trees, which becomes useful in our compiler as well.

\paragraph{A Spanning Tree.}
The proof contains of several fields, which will be explained one by one along with their roles.
The first field in the proof given to $u$ is its parent in the tree $\parent(u)$. This can be indicated by sending
$u$ the port number that points to its parent. Let $T$ be 
the graph defined by the $\parent(u)$ pointers. We must verify that $T$ is 
indeed a tree (\ie contains no cycles) and that it spans $G$. To verify that 
there are no cycles, the second field in the proof of $u$ contains the distance to the 
root in the tree, $d(u)$. The root should be given distance 0, and each node $u$ 
verifies with its parent in the tree that $d(u) = d(\parent(u))+1$. If there is 
a cycle in $T$, then no value for $d(u)$ can satisfy this requirement for all 
nodes on the cycle.
Finally, to be able to verify that $T$ \emph{spans} $G$, the third field in the proof is the 
ID of the root. Nodes verify with their neighbors that have the same root ID. If $T$ does 
not span $G$ then there must be two trees with two different roots. Since the graph is 
connected there must be an edge from one tree to the other which will spot the 
inconsistency of the root IDs.

The tree is used as a basic component in many protocols as it allows 
summing values held by each node (or computing other aggregative functions). For example, suppose we want to use compute $\sum_u x_u$ where $x_u$ is some number that is known to node $u$. Let $T_u$ be the 
subtree of $T$ rooted at $u$. We can use the prover to help us in this computation. Since the prover is untrusted, we will also need to verify this computation. This is done as follows.
The prover sends $u$ the value $X_u = \sum_{v \in T_u}x_v$, the sum of the $x_u$ values in the subtree $T_u$. Then, $u$ verifies that $X_u$ is consistent with the values given to its children in the tree. That is, $X_u=x_u+\sum_{i \in [k]} x_{v_i}$ where $v_1,\ldots,v_k$ are its children (the leaves have no children). If all values are consistent then we know that the root $r$ of the tree has the desired value $X_r = \sum_{u \in G}x_u$. We call such a procedure summing up the tree'' as it will be useful later on in different contexts.

\paragraph{A Reduction to Set Equality.}
Our main observation is that obtaining a general RAM compiler translates into
a specific problem of Set Equality. Let $\pi$ be a standard interactive protocol (with a centralized verifier). We construct a distributed protocol $\pi'$ as follows. First, we let the prover compute a spanning tree $T$ of the graph as described above, and assign IDs in the range of $1$ to $n$ for the nodes in the graph. The correctness of the spanning tree computation is verifies is in the labeling schemes described above.
We later describe how to also verify the correctness of the consecutive IDs in $[1,n]$. We will also solve this by a reduction to set equality.

The high level idea is to use the fact that the protocol is public-coin, and thus allows the prover to run the centralized verifier on its own. We now need the prover to convince the network that it simulated the verification algorithm correctly. For that purpose the verification of the RAM computation made by the prover is distributed among the $n$ nodes. Since the centralized RAM program consists of $O(n)$ steps, each vertex can be in charge of \emph{locally} verifying constant number of steps in this program. 
To verify that the computation is correct \emph{globally}, we will reduce the problem to Set Equality. 

We now explain it in more details.
Let $u_1,\ldots,u_n$ be the nodes ordered by their assigned IDs. 
Given this ordering, we can split the communication between the prover and verifier in $\pi$ to equally parts where node $i$ is responsible to communicate and store the responses of the $\ith{i}$ chunk of the messages. Since $\pi$ is a public coin protocol, the messages to the prover from each node are simply random coins. Finally, we need to simulate the verifier of $\pi$ by a distributed protocol. We assume that the verifier is implemented by a RAM program.

Consider a RAM program $C$. An execution of $C$ can be described as a sequence of read and 
write instructions to a
memory with $\poly(n)$ cells, where each operation consists of a short local state $s$ and a triplets $(v,a,t)$ where 
$v$ is the value (value read from memory or written 
to memory), $a$ is the address in the memory, $t$ is the timestamp of memory 
cell (\ie when it was last changed). We set the size of a cell to be $\log n$ bits such that each tuple the state $s$ and the triplets $(v,a,t)$ can be represented by $O(\log n)$ bits.

Let $R$ be the set of all read triplet operations and let $W$ be the set of 
all write triplet operations. Note that in general it might be that $R \ne 
W$, \eg if a cell is written once but read multiple times. Following the steps 
of \cite{BlumEGKN94} in the context of memory checking, we can transform any 
program to a {\em canonical form} where $R=W$ while paying only a constant factor in the running time. We assume from hereon that $C$ is given in this 
canonical form. Thus, we have that $R$ and $W$ describe an honest execution of 
the program if and only if $R=W$.

With this in mind, we can design the final step. Let 
$\timec$ be the running time of verifier $V$.
The prover runs the verifier and writes the list of triplets and local states 
of the program $C$. Each node is responsible for $\timec/n$ steps of the 
program, and the prover divides the triples and states of each instruction to 
the nodes. Each node that is responsible of step $i$ verifies that the states 
and triples are consistent with the instructions of step $i$ in the program 
$C$. What the node cannot verify locally is that the values read from the 
memory are consistent with the program. That is, we are left to verify is that 
the two sets $R$ and $W$ defined by the triplets are equal (as multisets). That 
is, we need a protocol for the problem $\SetEquality$.

\paragraph{A Protocol for Set Equality.}
As we have shown a protocol for Set Equality is the basis for the compiler. 
Actually, this protocol is used for other problems as well, and we describe it in its 
generality. Assume each node has an input $a_i$ and $b_i$ where $a_i,b_i$ are 
$O(\log n)$ long bit strings, and let $\A=\{a_i : i 
\in [n]\}$ and similarly $\B=\{b_i: i \in [n] \}$. We want to verify that 
$\A=\B$ as multisets.
We will describe here an $\dMAM[O(\log n)]$ protocol for this problem, which 
captures the mains ideas. In \Cref{sec:set-equality} we show how this protocol 
can be compressed to two message $(\dAM[O(\log n)])$, and we also show how to support each node holding two {\em lists} of up to $n$ elements (instead of single elements $a_i$ and $b_i$).

In the first message, we let the prover compute a spanning tree $T$ of the 
graph along with a proof as described above.
Then, to check that $\A=\B$ we define a polynomial $P_{\A}(x)$ and $P_{\B}(x)$ 
over a 
field $\F$ of size $n^3$ as follows:
$$
P_{\A}(x) = \prod_{u \in V}(a_u - x) \text{, and } P_{\B}(x) = \prod_{u \in 
	V}(b_u - x).
$$
As we show in the analysis, it holds that $\A = \B$ if and only if 
$P_{\A} \equiv P_{\B}$. Moreover, since the polynomials have low degree 
compared to the field size ($n$ vs.\ $n^3$) in order to check 
if they are equal it suffices to compare them on a random field element (if the 
two polynomials are different they can agree on at most $n$ element in the 
field).

We let the root of the tree $T$ sample a random field element $s \in \F$, and 
send it to the prover. The prover sends $s$ to all nodes of the graph. Nodes 
compare $s$ with their neighbors to verify that everyone has the same element 
$s$. Then, we are left with evaluating the two polynomials 
$P_{\A'}(s)$ and $P_{\B'}(s)$. To compute these polynomials we use a spanning 
tree $T$, and compute them ``up the tree''. We let the prover give each node 
$u$ the evaluation of the 
polynomials on the subtree $T_u$, that is, $A_u=\prod_{u \in T_u}(s-a_u)$ and 
$B_u=\prod_{u \in T_u}(s-b_u)$. Nodes check 
consistency with their children in the $T$ to assure that all partial 
evaluations are correct. That is, they check that 
$$
A_u=a_u \cdot \prod_{i \in [k]}A_{v_i} \text{, and } B_u= b_u \cdot \prod_{i 
	\in [k]}B_{v_i}
$$
where $v_1,\ldots,v_k$ are the children of $u$ in the tree.
Finally, the root $r$ of the tree holds the two complete evaluations of 
polynomials $A_r=P_{\A}(s)$ and $B_r=P_{\B}(s)$ and verifies that $A_r=B_r$. If 
all verifications pass then we know that with high probability $\A=\B$.

\paragraph{Assigning IDs.}
In the description above we assumed that unique IDs in the range 
of $1$ to $n$ are honestly generated. We show that this assumption is without loss of generality. Let the ID of node $i$ be $a_i$. Each node verifies that $a \le a_i \le n$ ($n$ is known to all nodes). We want to verify that the $a_i$ are all distinct. That is, we want to verify that the $a_i$'s are a permutation of $[n]$. This is also called the $\Permutation$ problem.

This 
is solved by reducing it to the Set Equality problem. Each node $i$ sets 
$y_i=a_i + 1 \mod n$. Let $\A = \{a_i: i \in [n] \}$ and $Y = \{y_i : i \in 
[n]\}$. Our key observation here is that the $a_i$'s are distinct if and only 
if $\A=Y$. 
Thus, we run the set equality protocol on $\A$ and $Y$. Note that this can be 
performed in parallel to the compiler's protocol, and thus does not add to the 
round complexity. 

\subsection{Asymmetry and Graph Non-Isomorphism}
We first give a short description of a {\em standard} (centralized) interactive protocol for $\ASym$ which is a simple adaptation of the public-coin protocol for graph non-isomorphism \cite{GoldwasserS89,GoldreichMW91} (see also \cite{BabaiM88}). From hereon we denote this protocol by the ``$\GNI$ protocol''. Then we show how to transform it to a distribute protocol.

Let $S$ be the set of all graphs that are isomorphic to $G$. That is, $S=\{G':G \cong G'\}$. The main observation of the $\GNI$ protocol which follows here directly is that if $G$ has no (non-trivial) automorphism then $|S| = n!$ while if $G$ does have an automorphism then $|S| \le n!/2$. Thus, the focus of the protocol is on estimating the size of $S$.

The verifier samples a hash function $g\colon \bit^{n^2} \to \bit^{\ell}$, where
$\ell$ is roughly $n\log n$ and sends it to the prover. The prover seeks for a graph $G' \in S$ such that $g(G')=0^{\ell}$. The main observation is that the probability that such a graph $G'$ exists is higher when $S$ is larger which allows the verifier to distinguish between the cases of $S$. That is, the verify will accept if $g(G')=0^{\ell}$.

Let us begin with an immediate solution for sparse graphs. Suppose that the graph $G$ is sparse (has $O(n)$ edges) and thus can be represented by $O(n\log n)$ bits. One can observe that in this case the total communication of the ``GNI protocol'' is linear in the input size, that is, $O(n\log n)$ and thus can be distributed among the nodes such that each node gets $O(\log n)$ bits. Finally, the verifier is required to compute the hash function $g(G')$. We need a very fast (linear-time) pairwise hash function for this. Luckily, Ishai et al.~\cite{IshaiKOS08} (see \Cref{cor:linear_ikos}) constructed such a hash function that can be computed in $O(n)$ operations over words of size $O(\log n)$. Thus, applying our RAM compiler with this hash function gives a $\dAMAM[O(\log n)]$ protocol for the problem: the first message is sending $g$ and messages 2-3 are sending $G'$ and verifying that $g(G')=0^{\ell}$.

The protocol above of course works only for sparse graphs as they had a small representation. While graphs in general have representation of size roughly $n^2$, since the size of the set $S$ is at most $n!$, any graph in $S$ can be indexed to have size $n^2$. Thus, we want to hash the set $S$ using a hash function $h$ to a set $S'$ such that $|S|=|S'|$ and each elements in $S'$ is represented using $O(n\log n)$ bits. While this approach is simple, it has a major caveat: computing $h(G)$ is exactly the task we wished to avoid! However, there is an important difference: the function $g$ had exactly $\ell \approx n\log n$ bits of output where $h$ has $cn\log n$ for a large constant $c$. This slackness in the constant lets us compose a special hash function $h$ that can be computed locally. Them, we will apply $g$ to the smaller elements of $S'$ and compute it using the RAM compiler as before. Together, we will verify that $g(h(G'))=0^{\ell}$.

In more details, our hash function $h$ will be composed of $n$ hash functions.
Each node $u$ chooses a seed for an $\epsilon$-almost-pairwise
hash function
$$
h_u\colon \bit^{n^2} \to \bit^{3\log n}
$$
where $\epsilon = 1/n$. The seed length of $h_u$ is $O(\log n)$ bits. Let
$h_1,\ldots,h_n$ be the $n$ chosen hash function ordered by the index of the nodes. Let
$G=x_1,\ldots,x_n$ where $x_i \in \bit^n$ is the indicator vector for the
neighbors of node $i$ in $G$. Then, we define a hash function $h \colon \bit^{n^2} \to \bit^{3n\log n}$ as
$$
h(G)=h_1(x_1) \concat \ldots \concat h_n(x_n).
$$
Using $h$ we can define the set $S' = \{h(G): G \in S\}$. It is easy to see that $|S'| \leq |S|$.
The fact that $h$ is locally computable means that it has a very bad collision probability. If two inputs are differ only on a single bit then the probability that they collide depend only on a single $h_i$ which rather small compared to the total range of $h$. To show that there will be no collisions under under $S$ we exploit the specific properties of $S$. The key point is that $S$ contains only graphs that are all isomorphic to each other and hence there are not many isomorphic graphs that differ only on a small part. This lets us bound the collision probability of two graphs as a function of their hamming distance $k$ and union bound over the number of isomorphic graphs of distance $k$. We show that with high probability we have that $|S|=|S'|$ and thus we can apply the protocol for $S'$ instead of $S$.

\paragraph{Graph Non-Isomorphism.}
The end result is a protocol for $\ASym$ in $\dAMAM[O(\log n)]$. In \Cref{sec:gni} we show how to adapt this protocol for $\GNI$, where we assume that in the $\GNI$ problem formulation nodes can communicate on both graphs $G_0$ and $G_1$.
We note that while this improves upon the $\dAMAM[O(n\log n)]$ of \cite{KolOS18}, our protocol works only when the $\GNI$ problem is defined such that nodes can communicate on {\em both} graphs $G_0$ and $G_1$. The protocol of \cite{KolOS18} works also on the definition $\GNI$ where only $G_0$ is the communication graph and $G_1$ is given as input nodes. That is, each node $u$ is given a list of its neighbors in $G_1$ but cannot communicate with them directly. This is not an issue when the proof complexity is $O(n\log n)$ as the prover can send each node $u$ its neighbors in the graph $G'$. However, when restricting the communication size to $O(\log n)$ this raises many difficulties, which seem hard to overcome.

\subsection{A Compiler for Small Space and Low Depth}
We describe how to get a compiler for small space computation (Item 1 in 
\Cref{thm:intro-small-space-depth}).
The main tool behind the construction is the interactive protocol of Reingold, 
Rothblum and Rothblum~\cite{ReingoldRR16}.
They show that for every statement that can be evaluated in polynomial time and 
bounded-polynomial space there exists a constant-round (public-coin) 
interactive protocol with an (almost) linear verifier. This is an excellent 
starting point for us, as our RAM compiler is most efficient for linear verifiers.

There is a subtle point here however. A linear-time in \cite{ReingoldRR16} is with respect to the size of the graph, \ie $m = 
O(n^2)$, whereas a linear time for our RAM compiler is with respect to the number of vertices $n$. 
To handle this, we first reduce the running time of the centralized verifier to 
$O(n)$ before applying our 
RAM compiler. Indeed, as already observed in \cite{ReingoldRR16}, the running time of the 
verifier can be made {\em sublinear} (\eg $m^{\delta}$ for some small constant 
$\delta$) if the verifier is given an {\em oracle access to a low degree extension 
of 
the input} (the input is the graph and possibly additional individual inputs 
held by each node). Our protocol will run the RAM-compiler on this sublinear version of the verifier while 
providing it this query access.
Luckily, evaluating a point of a low degree extension of the 
input is a task that is well suited for a distributed system, as it is a linear 
function of the input and hence can be computed ``up the tree'' using the 
prover. Thus, the \cite{ReingoldRR16} protocol can be compiled to a distributed 
one with constant number of rounds and $O(\log n)$ proof size.

A protocol with the same properties is given by Goldwasser, Kalai and 
Rothblumin~\cite{GoldwasserKR15} in the 
context of low depth circuits (as opposed to small space).
Let the class ``uniform $\NC$'' be the class of all language computable by a 
family of $O(\log(n))$-space uniform circuits of size $\poly(n)$ and depth 
$\polylog(n)$.
They showed the any 
languages computable by ``uniform $\NC$'' there is a public-coin interactive 
protocol where verifier runs in time 
$\polylog(n)$ given oracle access to a low degree extension of the input and 
the 
communication complexity is $\polylog(n)$.
Using the same approach as we did 
for the \cite{ReingoldRR16} protocol, we can also compile this protocol to a 
distributed one with polylogarithmic number of rounds and proof size.

\subsection{Below the $\log n$ Barrier}
To construct protocols with $o(\log n)$ proofs, we need to re-develop the basic ``distributed NP'' primitives only with a proof size in the required regime. Similar to the generality of the basic tree 
construction in distributed NP proofs, these tools are useful for many problems.

\paragraph{Constructing a Spanning Tree.}
We begin by showing how to compute a spanning tree in the graph using only 
$O(1)$ bits. We let the prover compute a BFS tree in the graph. 
However, the prover cannot even give a node $u$ 
its parent $\parent(u)$ in the graph, let alone prove its validity.

We take a different approach, using the specific properties of a BFS tree. 
If a node $u$ is in level $i$ in the BFS tree, then its neighbors are all in level $i-1$, $i$ or 
$i+1$. Thus, we let the prover give each node its distance from the root modulo 3. This 
gives each node sufficient information to divide its neighbors into three groups: 
neighbors in the same level $i$ as $u$, neighbors that are one level closer to the root, 
$i-1$, and neighbors that are one level below, $i+1$. The node $u$ defines its parent 
$\parent(u)$ to be its neighbors in level $i-1$ with the minimal port number (all neighbors 
of each node $u$ are ordered by an arbitrary port numbering that is known to the prover).
This way, each node $u$ has a defined parent $\parent(u)$ in the graph, except 
if it had no neighbors of level $i-1$ which means that it is the root.

Let $T$ be the graph defined by $(u,\parent(u))$. As in the standard proof 
labeling scheme for verifying a spanning tree, we first verify that $T$ is a 
tree (has no cycles), and that verify that it is also spanning.

First, we verify that there are no cycles in $T$. Towards this end, we let each 
node 
$u$ sample a uniform bit $b_u$ and send it to the prover. Let $P_u$ be the path in the 
tree that the prover computed from $u$ to the root.
The prover sends each 
node $u$ the number $s(u)=\sum_{v \in P_u}b_u \mod 2$, that is the sum of the $b_v$'s 
on the path from $u$ to the root modulo 2. Nodes exchange this value with their parent 
in the tree. Each node $u$ 
verifies that $s(u)=s(\parent(u)) + b_u \mod 2$. In the analysis, we show that if $T$ 
contains a cycle, then with 
probability $1/2$ the nodes will reject (this happens when the sum of the $b_u$ values 
on a cycle is odd).

By now, we know that $T$ contains no cycles. However, it might still be the case that $T$ 
is a forest. In such a case it will contain more than one root node. To eliminate this, we 
have 
the prover broadcast the value $b_r$ where $r$ is the root of the tree he computed. If 
there are more than one root in $T$, then with probability $1/2$
their $b_r$ values will be different and thus nodes will detect this inconsistency.
This insures that $T$ has no cycles and a single root thus it must be a spanning tree of 
$G$. Of course, the soundness can be 
amplified by standard (parallel) repetition.

A corollary of the constructing such a tree is that the root of the tree is a unique chosen 
node in the network. Thus, this protocol also solves the ``Leader Election'' problem 
($\LeaderElection$) with a constant size proof in 3-rounds.

\paragraph{Super Protocols.}
Our next step is to show how to run what we call ``super protocols''. A super 
protocol simulates a protocol with proof size $O(\log n)$ using only $O(\log 
\log n)$ bits, by making computation on a super graph $H$ that contains $n/\log 
n$ super-nodes. The super graph is defined by decomposing the graph into blocks 
of size roughly $\log n$ such that each block will simulate a single node in 
the protocol. The benefit of this approach is that a block has a proof capacity 
of $O(\log n)$ by having each node get only a single bit. In other words, a 
super-node (that corresponds to the block of $\log n$ nodes) can be given a 
proof of size $O(\log n)$ in a distributed manner: giving a single bit proof 
for each of node in that block.

This brings along several challenges as no node knows the $O(\log n)$ proof, 
but rather it is distributed among several nodes. To be able to work with these 
``fragmented proofs" we will need to come up with protocol that work on the 
super graph. Suppose a node $u$ in the super graph $H$ represents a block $B$. 
To simulate a local verification of $u$ in the super graph $H$, we need all 
nodes $B$ to cooperate to perform this verification. Towards this end, we will 
use the RAM compiler on a program $C$ that performs the verification, but we 
run the compiler only on the block $B$, as if it was the entire graph. Since 
the size of the block is $n'=O(\log n)$ the cost of this compiler is only 
$O(\log 
n') = O(\log \log n)$! Furthermore, the node $u$ performs consistency checks 
with its neighbors in $H$. Here again we use the RAM compiler, but on a graph 
that contains $u$ and a child $v$ of $u$. The graph of these two blocks is 
connected, and of size $O(\log n)$. This is carefully performed in parallel for 
all children $v$.

This was a very high level overview, and we proceed with formally explaining 
how to defines the blocks and the corresponding super graph.
The spanning tree $T$ (whose construction was described before) is partitioned into {\em edge-disjoint} subtrees $T_1,\ldots, T_k$, which we call \emph{blocks}. The precise protocol for this decomposition is given in \Cref{sec:block-decomposition}. The main point here is that at the end of the protocol, each node knows its neighbors within the block.

Using the block decomposition, we show how to reduce the proof size in the 
protocol for $\SetEquality$ to $O(\log \log n)$, albeit at the expense of more 
rounds.
The prover orders the nodes inside each block and sends each node its index $i$ 
inside the block. Since the blocks are of size $O(\log n)$ the index $i$ 
requires only $O(\log \log n)$ bits. To verify that the indexes are indeed a 
permutation, we apply the permutation protocol described above. However, we run 
it on each block separately as if the block was the whole graph. Since each 
block is 
of size $n'=O(\log n)$ the final cost of this protocol within each block is 
only $O(\log n') = O(\log \log n)$!

We wish to run this protocol in {\em parallel} for all blocks in the graph. 
This works if the blocks {\em vertex disjoint}, however, the block we have are 
only {\em edge disjoint}. Nodes that participate in several blocks will get a 
proof for each block which blows up the proof size. Instead, we show how such 
node get divide their proofs among the blocks. At the end, we are able to run 
the protocols in parallel without paying an additional cost for these nodes.

The next step of the $\SetEquality$ protocol, is to have the root choose a 
field element $s$ described by $O(\log n)$ bits. Let $r$ be the root of the 
tree $T$ 
and let $T_r$ be the block containing $r$. We let the {\em block} $T_r$ to 
distributively choose $s$, where each node picks a single bit. The 
prover reconstructs $s$ and can continue with the protocol. The main 
challenge now is that no individual node knows $s$, only the prover.

After $s$ has been chosen and sent to the prover, the next step of the protocol 
is to compute the products $\prod_{u \in G}(s-a_u)$ and $\prod_{u \in 
G}(s-b_u)$ and verify that they are equal. First, we compute each product 
within a block. Let $T_u$ be a block rooted at $u$, then we want the block to 
compute $a'_u=\prod_{u \in T_u}(s-a_u)$. Thus, we let the prover compute $a'_u$ 
and send it to the block $T_u$. To verify this, we can the RAM compiler on the 
block for a program $C$ that reconstructs $s$, computes $\prod_{u \in 
T_u}(s-a_u)$ and finally compares it to $a'_u$ (and similarly for the $b_u$'s). 
Again, this is performed for 
all blocks in parallel and has a cost of $O(\log \log n)$ bits.

Each node $u$ in the super graph $H$ now has the value $a'_u$, and we verified 
that $a'_u$ is indeed the product of all elements inside this block. Now, the 
prover computes the values $A_u = \prod_{u \in H_u}$ where $H_u$ is the subtree 
of $H$ rooted at $u$, and sends $A_u$ to the block $T_u$ (and similar for 
$B_u$). Now, node $u$ needs to verify this value by computing the product of 
$a'_v$ for all its children $v_i$.

We note that the block of $u$ and its children blocks are connected. 
Assume for simplicity, that $u$ has only a constant number of children blocks. 
Let $G'$ be the graph that contains all these blocks. Then, we have that $G'$ 
consists of $O(\log n)$ vertices. Thus, we run the RAM compiler on this graph, 
for a program $C$ that on input all the values of the nodes, collects the bits 
of $s$ and reconstructs it, then reconstructs $A_u$ and $A_{v_i}$ for all the 
children blocks, and verifiers $A_u = a'_u \cdot \prod_{i} A_{v_i}$. The size 
of the 
graph is $O(\log n)$ and thus again running this will cost $O(\log \log n)$ 
bits. 

This worked since we assumed that there are only a few child blocks, however 
the number of such blocks in general might be large. In such a case, we compute 
$\prod_{i} A_{v_i}$ by computing them in pairs $A_{v_i} \cdot A_{v_{i+1}}$, 
such that for each pair the graph is always of size $O(\log n)$. This takes 
some delicate care of details.
While this process is sequence and will take many iterations (as the number of 
children) we show how to parallel this using the prover. 

There many technical challenges to make this plan go through and we refer the 
reader to \Cref{sec:loglog} for the full details. The result is a five message 
protocol: first the prover sends the tree (and it is verified in messages 2-3), 
then the network chooses $s$ and then we run the RAM compiler in messages 3-5.

Once we have a protocol for $\SetEquality$ using $O(\log \log n)$ bits of 
 proof, we immediately get a protocol for $\DSym$. In this problem, the nodes 
 know a permutation $\pi$ and need to verify that it is an automorphism. We 
 simply run the $\SetEquality$ protocol on the two sets of edges for $G$ and 
 $G'=\pi(G)$.

\paragraph{A Protocol for $\Clique \in dMAM[O(1)]$.} 
We describe a protocol for the clique problem, where the goal is to prove that the graph 
contains a clique of size $K$ where $K$ is known to all. The prover marks a clique of size 
$K$ selects one of the nodes in the clique to be a leader. We run the leader protocol 
described above to verify that indeed a single leader is selected. Finally, each marked 
nodes verify that indeed $K-1$ of its neighbors are marked and that one of them is the 
leader. This assures that there are exactly $K$ marked nodes and that they form a clique.
\section{Definitions}
\subsection{Interactive Proofs with a Distributed Verifier}
Our definition follows the definition in \cite{KolOS18}.
An interactive proof is a protocol between a verifier and a powerful prover,
where the goal of the prover is to convince the verifier that $x \in L$ for
some common instance $x$ and language $L$. Usually, the verifier and prover are
turning machines with different computational power. Here, we consider the case
where the verifier is {\em distributed}.

Our model consists of a network of $n$ computation units that communicate in
synchronous
rounds. The communication pattern between the units is defined by an $n$-vertex
graph $G$. In additional, each node $u$ may hold an additional input $I(v) \in
\bit^{n}$. Let $I$ be the set of all inputs. Then, the graph $G$ and the inputs
$I$ define an instance $x$, and the goal of the network is to determine if $x
\in L$ for some language $\cL \subset \mathcal{G} \times I$, where
$\mathcal{G}$ is a family of $n$ vertex graphs and $I$ is a set of inputs where $I(u)$ is the input of node $u$.

The network is equipped with an one extra entity, $\prover$, which we call the
\emph{prover}. This prover is connected to all the vertices in $G$, and knows
the entire input instance $\langle G,I \rangle$. Roughly speaking, the goal of
this powerful prover is to convince the network that $x \in L$, where if $x
\notin L$ we ask that the network will not be convinced no matter what the
prover does. The prover knows the entire graph: it knows the ordering or the neighbors for each node $u$ in the graph.

\paragraph{The Complexity Measures.}
Our
primary goal in this paper is to minimize the bandwidth, that is, the size of messages sent in each round (within the network and also between the nodes and the prover). The total amount of messages sent is called the proof size (or proof complexity) of the protocol.

\paragraph{The class $\dIP[r,\ell]$:}
Let $\cL$ be a language of graphs and inputs and let $r,\ell$ be two
parameters. For a verifier $\verifier$ and a prover $\prover$ we let $\langle \verifier,\prover \rangle$ denote the protocol between them and we let $\VerifierAlg^{out}(u)$ be final output of the vertex $u$ in the protocol. 
We say that $\cL \in \dIP[r,\ell]$ if there exists an $r$-round
protocol (\ie $r$ messages) with verifier $\VerifierAlg$ with the following properties:
\begin{enumerate}
	\item \textbf{Completeness:} For every $\langle G, I \rangle \in \cL$, there exist a prover $\ProverAlg$ such thats for $\langle \verifier,\prover \rangle$ it holds that
	$\Pr[\forall u, \VerifierAlg^{out}(u)=1] > 2/3$.
	\item \textbf{Soundness:} For every $\langle G, I \rangle \notin \cL$ and every prover $\ProverAlg^*$, we have for $\langle \verifier,\prover \rangle$ it holds that 	$\Pr[\forall u, \VerifierAlg^{out}(u)=1] < 1/3$.
\end{enumerate}
The probabilities are taken over the random coins of the nodes of the
distributed verifier $\verifier$ in the protocol between verifier and the prover $\langle \verifier,\prover \rangle$.

When $r=1$ and prover goes first, this is the standard notion of distributed proofs (or proof labeling schemes). When $r=2$ the verifier sends the first message this is the analog of the $\AM$ calls and denoted as $\dAM[\ell]$. Similarly, we define $\dMAM[\ell]$ for three rounds and $\dAMAM[\ell]$ for four message and so on.

\subsection{Limited Independence}
A family of functions $\cH$ mapping domain $\bit^n$ to range $\bit^m$ is
$\epsilon$-almost pairwise independent if for every $x_1 \ne x_2 \in \bit^n$,
$y_1,y_2 \in \bit^m$, we have
\begin{align*}
\Pr_{h \in \cH}[h(x_1) = y_1 \wedge h(x_2) = y_2] \le \frac{1+\epsilon}{2^{2m}}.
\end{align*}

\begin{theorem}\label{thm:almost-pairwise}
There exists a family $\cH$ of $\epsilon$-almost pairwise independent
	functions from $\bit^n$ to $\bit^m$ such that choosing
	a random function from $\cH$ requires $O(m + \log n + \log(1/\epsilon))$
	bits.
\end{theorem}

\paragraph{Circuits.}
Some of our results used the notions of circuit. In this work, we consider circuits of constant fan-in and fan-out. The term ``linear size'' circuits refers to circuits whose size is linear in the sum of their input size and output size.

\paragraph{Linear Hash Functions.}
Ishai et al.~\cite{IshaiKOS08} showed how to construct a pairwise independent hash
function that can be computed by a linear-sized circuit. Specifically:

\begin{corollary}\cite[Follows from Theorem 3.3]{IshaiKOS08}
\label{cor:linear_ikos}
Let $\F$ be a field of size $n$.
There exists a family $\cH$ of pairwise independent hash functions from
$\F^{n}$ to $\F^{n}$ such that choosing a random function from
$\cH$ requires $O(n)$ field elements and evaluating any $h \in \cH$ can be
performed by an $O(n)$-sized circuit with gates that operate over $\F$.
\end{corollary}


\begin{definition}[Aggregate Function]\label{def:aggregate-function}
	We say that a function $f : \bit^{n \times m} \to \bit^n$ is an aggregate
	function if there exists a function $g : \bit^{2n} \to \bit^n$ such that
	$f(x_1,\ldots,x_m) = y_m$ where $y_i = g(x_i,y_{i-1})$ for $2 \le i \le m$ and
	$y_1 = x_1$, and $g$ is computable in $O(n)$ by a RAM program with operations over words of length $w=O(\log n)$.
\end{definition}

\subsection{Graph Definitions}
We usually denote the graph by $G=(V,E)$ where $V$ is the set of vertices and $E$ is the set of edges. We let $N(u)=N_G(u)$ denote the neighborhood of $u$ in $G$. We also call the vertices in $V$ nodes.

\begin{definition}[Isomorphism]
 We say that two graphs $G=(V,E)$ and $G'=(V',E')$ are isomorphic if there exists a bijection $\pi$ between $V$ and $V'$ such that for any two nodes $u,v$ it holds that $(u,v) \in E$ if and only if $(\pi(u),\pi(v)) \in E'$. We denote this by $G \cong G'$.
\end{definition}

\begin{definition}[Automorphism]
A graph $G=(V,E)$ has an automorphism if there exists a non-trivial permutation $\pi$ such that for every $u,v \in V$ it holds that $(u,v) \in E$ if and only if $(\pi(u), \pi(v)) \in E$ (we call such a graph Symmetric).
\end{definition}

%
\section{A RAM Program Compiler}
In this section we show our RAM program compiler. We take standard interactive protocols over $n$-vertex graphs and transform them into distributed protocols. The cost of the distributed protocol depends on the running time of the verifier in the protocol when implemented as a RAM program.

A construction of a spanning tree in the graph $G$ is a basic tool in distributed proofs in general \cite{KormanKP10} and in our context as well. Here, we let the prover compute a spanning tree $T$ rooted at an arbitrary  node $r$ and send each node its parent $\parent(u)$ in the tree (the parent of the root is $\bot$). Note that once each node knows its parent in the tree, it also knows its children in the tree.

Then, to prove that this is 
indeed a tree, the prover additionally gives each node its distance from the 
root, $d(u)$ in the tree $T$. Each node verifies consistency with its parent, 
\ie $d(u) = d(\parent(u))+1$ (the root $r$ verifies that $d(r)=0$). One can observe that verifying the distances from the root 
assures that there are no cycles in $T$ as otherwise there must be a node $u$ and its parent $\parent(u)$ with inconsistent distances. Finally, to prove that the tree is spanning the prover gives each node the ID of the root where nodes verify consistency of the ID with their neighbors.

Using this tree, we develop an {\em interactive} protocol for a new problem we 
call $\SetEquality$ (defined next). This protocol will be used several times in 
our compiler (and later on) and in particular is used in a protocol for the 
$\Distinctness$ problem and $\Permutation$ program (also defined next). Next, 
we describe the $\SetEquality$ problem.

\subsection{$\SetEquality \in \dAM[O(\log n)]$}\label{sec:set-equality}
The $\SetEquality$ equality checks the equality of two (multi)sets and is formally defined as follows.
\begin{definition}[$\SetEquality$]
In this problem each node $u$ holds two lists of $\ell$ elements $a_{u,1},\ldots,a_{u,\ell}$ and $b=b_{u,1},\ldots,b_{u,\ell}$ where for all $i \in [\ell]$ it holds that $a_{u,i},b_{u,i} \in \bit^{c\log n}$ for some constant $c \in \N$ and $\ell \le n$. Let $\A=\{a_{u,1} : u \in V, i \in [\ell]\}$ and $\B=\{b_u : u \in V, i \in [\ell]\}$ be two {\em multisets}. The goal of the $\SetEquality$ problem is to prove that $\A=\B$ as multisets. 
\end{definition}
Let $G=(V,E)$ be an $n$-vertex graph and let $\F$ be a field of size $n^{c + 3}$.
We interpret the elements of $\A$ and $\B$ as elements in the field $\F$.
To check that $\A=\B$ (as multisets) we define a polynomial $P_{\A}(x)$ and $P_{\B}(x)$ 
according to the elements of $\A$ and $\B$ respectively. That is, we define
$$
P_{\A}(x) = \prod_{u \in V, i \in [\ell]}(a_{u,i} - x) \text{, and } P_{\B}(x) = \prod_{u \in V, i \in [\ell]}(b_{u,i} - x).
$$
Note that $P_{\A}$ and $P_{\B}$ are polynomial of degree at most $n\ell$. We 
show that $\A = \B$ if and only if $P_{\A} \equiv P_{\B}$. Since the 
polynomials have low degree (compared to the field size), in order to check if 
they are equal it suffices to compare them on a random field element. For clarity of presentation, let us assume that nodes have shared randomness. At 
the end, we show to sample this shared randomness using the prover.

Thus, let $s \in \F$ be a random field element defined from the shared 
randomness. Then, we are left with evaluating the two polynomials 
$P_{\A}(s)$ and $P_{\B}(s)$. To compute these polynomials we use a spanning 
tree construction, as described above. We let the prover compute a spanning tree $T$ and prove its validity.
We use the tree $T$ to compute the two polynomials on $s$. 
Towards this end, the prover sends each node $u$ the evaluation of the 
polynomials on the subtree $T_u$: $A_u=\prod_{u \in T_u, i \in [\ell]}(s-a_{u,i})$ and $B_u=\prod_{u \in T_u, i \in [\ell]}(s-b_{u,i})$. Nodes check 
consistency with their children in the $T$ to assure that all partial 
evaluations are correct. That is, they check that 
$$
A_u=\prod_{i \in [\ell]}a_{u,i} \cdot \prod_{j \in [k]}A_{v_i} \text{, and } B_u= \prod_{i \in [\ell]}b_{u,i} \cdot \prod_{j 
\in [k]}B_{v_i}
$$
where $v_1,\ldots,v_k$ are the children of $u$ in the tree.
Finally, the root $r$ of the tree holds the two complete evaluations of 
polynomials $A_r=P_{\A}(s)$ and $B_r=P_{\B}(s)$ and verifies that $A_r=B_r$.

This completes the description of the protocol assuming the element $s$ is shared randomness. To construct such shared randomness we do the 
following. We let each node $u$ 
sample $s_u$ at random, along with a random number $\alpha_u \in [n^2]$. The 
node $u^*$ with the minimal $\alpha_{u^*}$ ``wins'' in terms that we set $s = 
s_{u^*}$ and $\alpha=\alpha_{u^*}$ (observe that we cannot have the prover 
decide who wins, as otherwise $s$ could be biased). The prover will announce to everyone the winning $\alpha$ and $s$. Nodes verify the consistency of $s$ and $\alpha$ with their neighbors, and thus assure that a nodes in the graph has the exact same elements $s$ and $\alpha$. We are left to verify that indeed $\alpha$ is the minimal one value.

To verify this, each node $u$ will check that indeed $\alpha \le 
\alpha_u$ where we expect exactly a single node $u^*$ to have equality. We 
count the number of such nodes by having the prover send each node $u$ the 
number of nodes that have equality in its subtree. That is, the prover sends node $u$ the value $Q_u = \sum_{v \in T_u}I_{\alpha = \alpha_u}$ where $I_{\alpha = \alpha_u}=1$ if $\alpha = \alpha_u$ and 0 otherwise. The nodes check consistency of the $Q_u$ with their children in the tree and finally the root $r$ verifies that $Q_r=1$. This assumes a common random string $s$.
The formal protocol is given in \Cref{fig:protocol-setquality}. 

\protocol
{A protocol for set equality.}
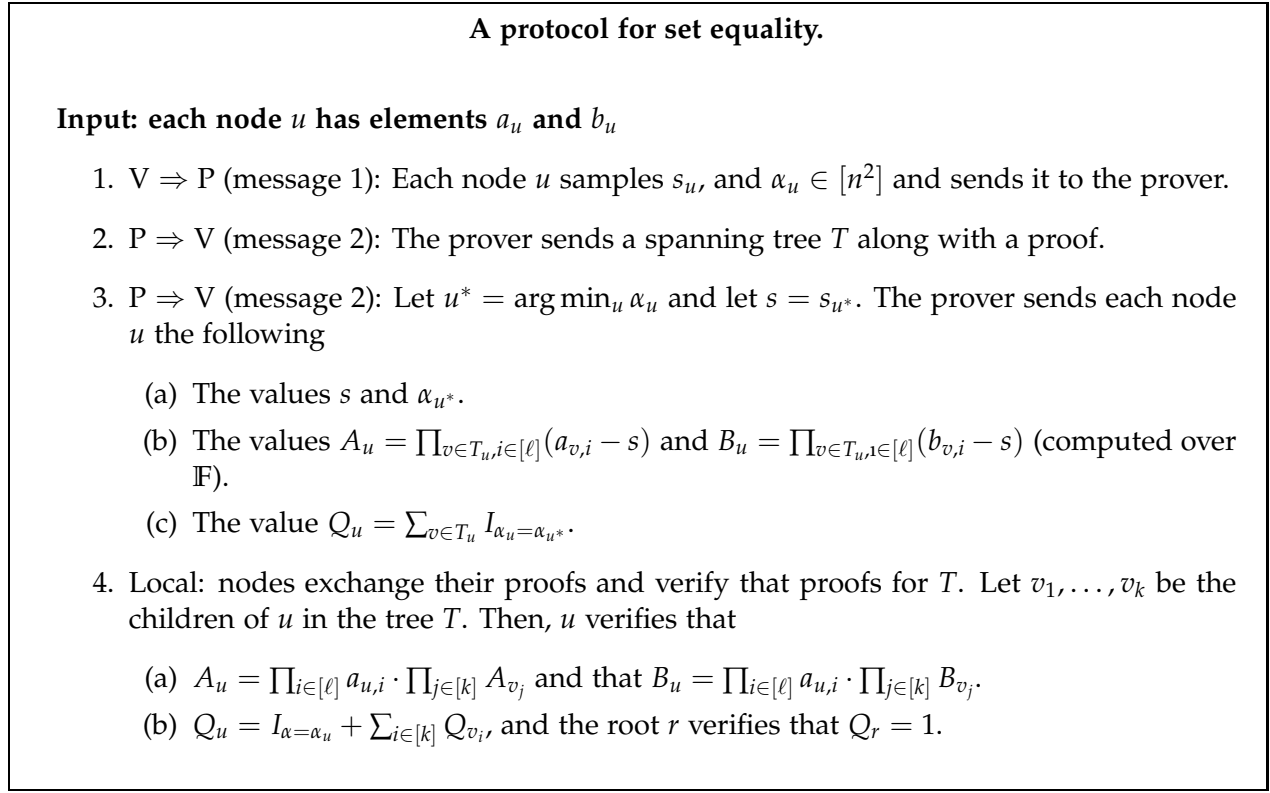
{A distributed AM protocol for checking the equality of two 
	multi-sets.}
{fig:protocol-setquality}
{
\vspace{3mm} \textbf{Input: each node $u$ has elements $a_u$ and $b_u$}
\begin{enumerate}
	\item V $\Rightarrow$ P (message 1): Each node $u$ samples $s_u$, and 
	$\alpha_u \in [n^2]$ and sends it to the prover.
	\item P $\Rightarrow$ V (message 2): The prover sends a spanning tree $T$ along with a proof.
	\item P $\Rightarrow$ V (message 2): 
	Let $u^* =\arg \min_u \alpha_u$ and let $s=s_{u^*}$. The prover sends each node $u$ the following
	\begin{enumerate}
		\item The values $s$ and $\alpha_{u^*}$.
		\item The values $A_u=\prod_{v \in T_u, i \in [\ell]}(a_{v,i} - s)$ and 
		$B_u=\prod_{v \in T_u, \i \in [\ell]}(b_{v,i} - s)$ (computed over 
		$\F$).
		\item The value $Q_u = \sum_{v \in T_u}I_{\alpha_u = \alpha_{u^*}}$.
	\end{enumerate}
	\item Local: nodes exchange their proofs and verify that proofs for $T$. Let $v_1,\ldots,v_k$ be the children of $u$ in the tree $T$. Then,  
	$u$ verifies that
	\begin{enumerate}
		\item $A_u=\prod_{i \in [\ell]}a_{u,i} \cdot \prod_{j \in [k]}A_{v_j}$ and that $B_u= \prod_{i \in [\ell]}a_{u,i} 
		\cdot \prod_{j \in [k]}B_{v_j}$.
		\item $Q_u = I_{\alpha=\alpha_u} + \sum_{i \in [k]}Q_{v_i}$, and 
		the root $r$ verifies that $Q_r = 1$.
	\end{enumerate}
\end{enumerate}
}

We show correctness and soundness of the protocol.
\paragraph{Correctness.}
The protocol succeeds as long as the $\alpha_{u^*}$ is uniquely the minimal 
value. However, it is easy to see that $\Pr[\exists i \ne j : \alpha_i = 
\alpha_j] \le 1/n$. Thus, we continue the analysis as if all the $\alpha_i$'s 
are distinct.
Assume that $\mathcal{A} = \mathcal{B}$ as multisets. 
Then for any $s \in \F$ it holds that $\prod_{u \in V, i \in [\ell]}(a_{u,i}-s) = \prod_{u \in V, i \in [\ell]}(b_{u,i}-s)$. For any tree $T$ with root $r$ it holds that $A_r = \prod_{u \in V,i \in [\ell]}(a_{u,i}-s) = \prod_{u \in V,i \in [\ell]}(b_{u,i}-s) = B_r$ and also that $Q_r=1$. Thus, the root 
$r$ will output 1, 
and in addition all intermediate nodes will output 1 after their local 
verification.

\paragraph{Soundness.}
Assume that $\mathcal{A} \ne \mathcal{B}$ as multisets. 
Suppose that $\prod_{u \in V,i \in [\ell]}(a_{u,i}-s) \ne \prod_{u \in V,i \in [\ell]}(b_{u,i}-s)$. In order for the prover to cheat, 
it must give the root $r$ values $A^*_r,B^*_r$ such that either $A^*_r \ne A_r$ 
or $B^*_r \ne B_w$, since otherwise the node $r$ will output 0. However, since the node $r$ performs the local check with its neighbors in the tree, it holds that the prover must give wrong values to one of its children as well. This 
continues until the prover gives a wrong value to a leaf, where the leaf can 
verify locally and output 0 indicating that it revived a wrong proof.

Thus, we bound the probability that the two products collide (notice that the 
sets are fixed before the choice of $s$). Consider the polynomial 
$f(x)=\prod_{u \in V,i \in [\ell]}(a_{u,i}-x) - \prod_{u \in V,i \in [\ell]}(b_{u,i}-x)$, which is of degree at most $n$ over the field $\F$.
\begin{claim}
	$f$ is not the zero polynomial.
\end{claim}
\begin{proof}
We know that $\A \ne \B$.
Suppose that there exists an element $z \in \A \setminus \B$. Then, we get 
that $\prod_{u \in V,i \in [\ell]}(a_{u,i}-z)=0$ and $\prod_{u \in V,i \in [\ell]}(b_{u,i}-z) \ne 0$, therefore 
$f(z) \ne 0 $ and thus $f$ is not the zero polynomial. A similar arguments 
holds if $z \in \B \setminus \A$. Since $\A$ and $\B$ are {\em multisets} 
there is a third possibility that the multisets share the same elements 
only with 
different multiplicities. Let $C$ be the multiset of their intersection $\C = 
\A \cap \B$. Define
$$
g(x)=f(x)/\prod_{c \in \C}(c-x).$$
It suffices to show that $g$ is not the zero function. Define $\A' = \A 
\setminus \C$ and $\B' = \setminus \C$. For these subsets we know that there 
must be an element that is in one set and not in the other. Assume without loss 
of generality that there must exist an element $z \in A' \setminus \B'$. 
Then, 
since $z \notin \C$ we get that $g(z) \ne 0$ and therefore $f$ is not the zero 
polynomial.
\end{proof}
The polynomial $f$ has at most $n\ell \le n^2$ roots and since the field is of size 
$n^{c+3}$ we get that
\begin{align*}
\Pr_r\left[\prod_{u \in V,i \in [\ell]}(a_{u,i}-s) = \prod_{u \in V,i \in [\ell]}(b_{u,i}-s)\right] = \Pr_r[f(s)=0] \le \frac{n\ell}{n^{c+3}} \le \frac{1}{n}.
\end{align*}

\paragraph{Communication Complexity.}
Computing the tree $T$ and its proof take $O(\log n)$ proof size, as shown in 
\cite{KormanKP10}. Elements in the field $\F$ are represented using $O(\log n)$ 
and each node is given a constant number of elements ($s,A_u,B_u$). We have 
that $\alpha \in [n^2]$ and thus also has short representation. The $Q_u$ are a constant number of bits. Altogether, each node sends and receives $O(\log n)$ bits.

\subsection{Distinctness}
In the $\Distinctness$ problem each node has a single value $a_i$ and the goal to verify that all values are distinct. That is, the output of the protocol is 1 if and only if it holds that for all $i \ne j$ such that $i,j \in [n]$ we have that $a_i \ne a_j$.

We show that this problem can be actually reduced to the $\SetEquality$ problem. Assume that the values $a_i$ are sorted such that $a_1 \le \ldots \le a_n$. The prover sends node $i$ the value $a_{i+1 \mod n}$. Denote by $y_i$ the actual value received by a node $i$. Then, node $i$ sets a bit $b_i$ to be 1 if and only if $a_i < y_i$. 

Let $A$ be all the original values $A=\{a_i\}_{i \in [n]}$ and let $Y$ be the set of all values given by the prover. Then, we run the protocol for $\SetEquality$ to verify that $A=Y$. Moreover, we run a sum protocol to verify that $\sum_{i=1}^{n}b_i = 1$. 

\protocol
{A protocol for distinctness.}
{A distributed MAM protocol for checking that each node in the graph has a 
unique identity 
$i \in [n]$.}
{fig:protocol-distinctness}
{
	\vspace{3mm} \textbf{Input: each node $i$ has a value $a_i$, and assume $a_1 \le a_2 
	\le \ldots \le a_n$.} \\
	\textbf{Output: 1 if and only if all values are distinct.}
	\begin{enumerate}
		\item P $\Rightarrow$ V (message 1): prover gives node $i$ the value 
		$a_{i+1 \mod n}$. Let $y_i$ be the received value.
		\item Local: node $i$ sets $b_i = 1$ if and only if $a_i < y_i$.
		\item P $\Rightarrow$ V (message 1): prover sends a proof that 
		$\sum_{i=1}^{n}b_i=1$.
		\item P $\Leftrightarrow$ V (messages 2-3): prover and verifier 
		interact to assure that $A=Y$, where $\A=\{a_i\}$ and $Y=\{y_i\}$.
	\end{enumerate}
}

\paragraph{Completeness.}
If all values $a_i$ are distinct then the honest prover will set $y_i=a_{i+1 \mod n}$. Thus, we will have that $b_i=0$ for all $i < n$ and $b_n=1$ and therefore $\sum_{i=1}^{n}b_i=1$. Moreover, we have that the values of $Y$ are exactly the values of $A$ shifted by 1. That is, as sets we have that $A=Y$ and thus the $\SetEquality$ protocol will pass as well.

\paragraph{Soundness.}
To show soundness we define an $n$-vertex directed graph with nodes being $A \cup Y$. Since the $\SetEquality$ protocol have passed successfully, we know that that $A=Y$ as multi-sets. It therefore holds that for any $i$ there exists a $j$ such that $y_i=a_j$. We then add the directed edge $(i,j)$ to the graph.

By the construction, we get that the in-degree and the out-degree of each node in this graph are exactly the node's multiplicity in $A$, which is at least 1. Thus, by an Euler argument, the graph can be decomposed to edge-disjoint cycles. However, since $\sum_{i=1}^{n}b_i=1$ we know that all but one edge are strictly increasing in values. Thus, the decomposition can contain only a single cycle, and thus the in-degree and out-degree are exactly 1, which means that the values of $A$ are all distinct.

\paragraph{The Permutation Problem.}
A specific instance of the $\Distinctness$ problem is when for $i \in [n]$ we 
have that $1 \le a_i \le n$. In such a case, we are actually checking if the 
sequence $a_1,\ldots,a_n$ form a permutation. We denote this problem by 
$\Permutation$.

We observe that in this problem, the first message from the prover 
is redundant: node $i$ can compute by itself the value $y_i=a_{i+1} \mod n = 
a_i+1 \mod n$. Thus, we only need to run messages 2-3 of the $\Distinctness$ 
protocol which yields an $AM$ protocol with the same proof complexity. That is, 
we have that $\Permutation \in \dAM[O(\log n)]$. The formal protocol is given 
in \Cref{fig:protocol-permutation}.

\protocol
{A protocol for $\Permutation$.}
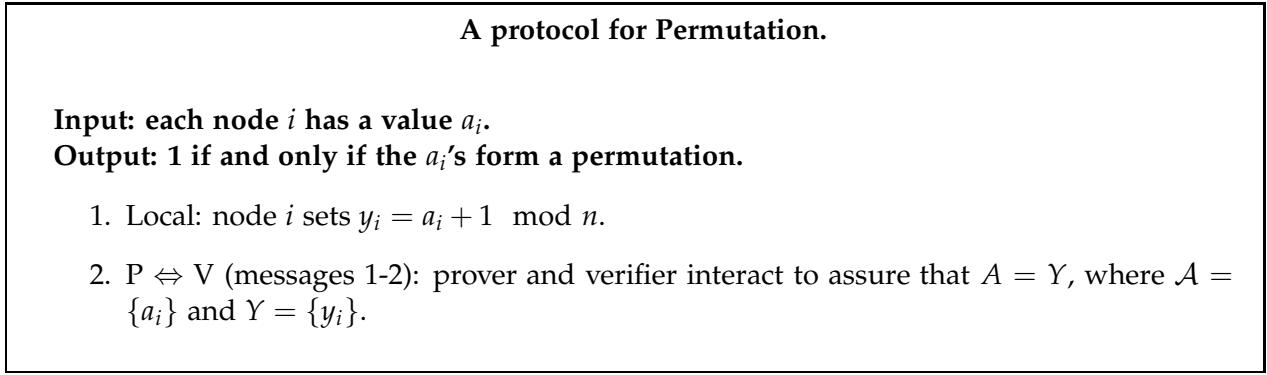
{A distributed AM protocol for checking a permutation.}
{fig:protocol-permutation}
{
	\vspace{3mm} \textbf{Input: each node $i$ has a value $a_i$.} \\
	\textbf{Output: 1 if and only if the $a_i$'s form a permutation.}
	\begin{enumerate}
		\item Local: node $i$ sets $y_i = a_i +1 \mod n$.
		\item P $\Leftrightarrow$ V (messages 1-2): prover and verifier 
		interact to assure that $A=Y$, where $\A=\{a_i\}$ and $Y=\{y_i\}$.
	\end{enumerate}
}
\subsection{The Compiler}
\label{sec:RAM_compiler}
We present a general compiler that takes standard interactive protocols and 
transforms them into distributed interactive protocols. Let $\pi \in \IP$ be an 
interactive protocol.

We show how to constructed the distributed version $\pi'$ of $\pi$ for an 
$n$-vertex graph $G$ with input $I$. First (message 1), we let the prover give unique IDs in the range of $1$ to $n$ to the nodes. This is verified using the $\Distinctness$ protocol (in parallel to messages 2-3). This lets us order the nodes $u_1,\ldots,u_n$ according to their assigned IDs.

Once the nodes are ordered, we can split the communication between the verifier and the prover to small parts for each node. Suppose that in the protocol $\pi$ the verifier sends a message $R$ to the prover. Since the protocol is public-coin, we know that $R$ is simply a random string. Thus, in 
the distributed version each node $u$ will send the prover a small random 
string, $R_u$, where $|R_u| = |R|/n$. The prover collects all $R_u$ and 
composes the random string according to the order of the nodes $R=R_1,\ldots,R_n$. 
Then, in the protocol $\pi$ the prover responds with a message $Y=\pi((G,I),R)$. In the distributed version $\pi'$ the prover distributes the 
string $Y$ among the $n$ nodes. Each node $u$ gets $Y_u$ where $|Y_u| = |Y|/n$ and $Y=Y_1,\ldots,Y_n$. This continues for all rounds of the protocol $\pi$.
If the total communication complexity of $\pi$ is $O(n)$ then 
the communication per node in $\pi'$ so far in the protocol is $O(1)$.

Let $\verifier$ be the verifier in the protocol $\pi$. At the end of the protocol, the verifier has pairs $(R,Y)$ for each round of the protocol distributed among the nodes. Let $\vec{R}$ be the collection of all the $R$'s and let $\vec{Y}$ be the collection of all the $Y$'s. According to $\pi$ in order to decide whether to accept we need to compute $\verifier((G,I),\vec{R},\vec{Y})$. However, computing this in a distributed manner is challenging as each 
node has a different part of the input to the $\verifier$'s program.

Here we let the prover help us in computing $\verifier((G,I),\vec{R},\vec{Y})$ 
by a three message protocol which we describe next. If the running time of 
$\verifier$ is $\timec$ then the 
communication complexity (per node) of the final protocol will be $\timec\log n 
/n$.
In general, our compilers takes as input a description of any $r$-round 
$\IP$ protocol where the computation of the verifier can be done a time 
$\timec$ 
by a RAM program, and transforms it to a distributed $r+2$-round protocol with 
proof complexity of $O(\timec/n\log n)$. In particular, for $O(n)$ time 
programs the proof size if $O(\log n)$.

\paragraph{A Canonical Form for RAM Programs.}
A RAM program is modeled as a CPU that has a small state (\eg containing the context of the registers) and performs a sequence of instructions to an external memory (the input to the program is assumed to be stored in the memory). Each instruction operations (\ie can read and write) on the local state and on a single cell in the external memory. The instructions are numbered and we say that instruction number $i$ happened on time $i$. Without loss of generality, we have each memory cell contain a timestamp of the last time it was updated. That is, if at time $i$ the program writes to memory address $a$ then the timestamp in cell $a$ will be updated to $i$.

Observe that using this formation, the set $\mathcal{W}$ of (value, address, time) triples which are written is equal to the set $\mathcal{R}$ of (value, address, time) that are read. However, $\mathcal{W}$ and $\mathcal{R}$ might differ as {\em multisets}, as if the program writes a value once and then reads it multiple times. We follow the footsteps of Blum et al.~\cite{BlumEGKN94} and show that any program can be easily transformed into a canonical form where $\mathcal{W}$ and $\mathcal{R}$ are equal as multisets as well, while paying only a constant factor in the running time. In short, we make the program read any location after writing it, and vice versa. Formally, we replace the read and write operations with the follows.

\noindent\textbf{Write of value $v$ to address $a$ at time $t$ with state $s$:}
\begin{enumerate}
	\item read value $v'$ and time $t'$ stored at address $a$.
	\item write value $v$ and time $t$ to address $a$.
	\item update state $s$.
\end{enumerate}

\noindent\textbf{Read address $a$ at time $t$ with state $s$:}
\begin{enumerate}
	\item read value $v'$ and time $t'$ stored at address $a$.
	\item write value $v'$ and time $t$ to address $a$.
	\item update state $s$.
\end{enumerate}

In our setting, the state $s$ and timestamps are each of size $O(\log n)$ and 
the memory is of size $\poly(n)$. Thus, a tuple $(s,v,a,t)$ (corresponding o 
the state, value to read/write, memory address and timestamp) can be described 
using $O(\log n)$ bits, and the execution of a $C$ can be described by a list 
of $\timec$ such tuples.

%
%

\paragraph{The Compiler.}
We let the prover run the computation of the program to get the description of 
its execution, \ie a list of $\timec$ tuples $\{(s,v,a,t)_i\}_{i \in 
[\timec]}$. We divide the steps among the $n$ nodes, such that each node is 
assigned $\timec/n$ arbitrary steps of the execution. If the output of the 
program $y$ is a Boolean value, then an arbitrary node $v^*$ that is assigned 
the last instruction of the program will have the final output $y$.
Denote by $I_u$ the set of step numbers that are assigned to node $u$. Then the prover sends node $u$ the tuples $\{(s,v,a,t)_i\}_{i \in [I_u]}$.

Our goal now is to verify that this is an honest execution of the program. Define $R$ to be the set of all $(v,a,t)$ for all the read operations and let $W$ be defined similarly for the write operations. Each node holds a list of tuples where each tuple is either in $R$ or in $W$. Thus, to verify that the sets are equal we run the $\SetEquality$ protocol as described in \Cref{sec:set-equality}.

Then, we check that the addresses and write values correspond to the 
instructions of the program. Let $(s,v,a,t)$ be the tuple corresponding to 
instruction $i$ of the program given to node $u$, and suppose it is a write 
instruction. Then $u$ runs the $\ith{i}$ instruction of the program $C_i$ on 
state $s$ to get address $a'$, value $v'$ and new state $s'$, that is, 
$C_i(s)=(s',a',v')$. Then ,$u$ checks that indeed $a'=a$ and $v'=v$ (for a read 
instruction we check only $a'=a$). Denote by $S$ the set of all pairs $(s,i)$ 
given to the nodes and denote by $S'$ the set of all pairs $(s',i+1 \mod 
\timec)$ where $s'$ is the state computed above (if $s$ is the state of the 
last instruction then we let $s'$ be the state of the first instruction). Then, 
are we need to verify that $S=S'$ and again use the protocol for $\SetEquality$ 
for this.
The protocol is given in \Cref{fig:protocol-ram-compiler}.

\protocol
{A Compiler for RAM Programs.}
{A distributed compiler for RAM programs}
{fig:protocol-ram-compiler}
{
\begin{enumerate}
	\item P $\Leftrightarrow$ V (messages 1-3): prover and verifier interact
	to 	establish unique IDs for the nodes in $[n]$.	
	\item P $\Rightarrow$ V (message 1): prover sends node $i$ the tuples
	$(v,a,t,s)_j$ for all $j \in I_i$. Let $\mathcal{R}$ be the set of read operations
	and let $\mathcal{W}$ the set of write operations (from all nodes).
	\item P $\Leftrightarrow$ V (messages 2-3): prover and verifier interact to
	prove set equality of $\mathcal{R}=\mathcal{W}$.
	\item Local: node $u$ with tuple $(s,v,a,t)_j$ computes $C_j(s)=s',v',a'$ 
	and verifies that $a'=a$ (and that $v'=v$ if $C_j$ is a write instruction). 
	Let $S=\{(s,j)\}$ be the set of all such states and let $S'=\{(s',j+1 \mod 
	n)\}$.
	\item P $\Leftrightarrow$ V (messages 2-3): prover and verifier interact to
	prove set equality of $S=S'$.
	\item Local: node $v^*$ with the final output $y$ of the program verifies
	that $y=1$.
\end{enumerate}
}

\paragraph{Completeness.}
The completeness follows directly from the construction. An honest prover will
provide true IDs for the nodes, follow the computation of the RAM program each
provide each node with the correct memory values read by the computation. Then,
we will have that $\mathcal{R}=\mathcal{W}$ and thus the node $v^*$ will accept.

\paragraph{Soundness.}
If the prover did not provide unique IDs in the range of $[n]$ to the nodes,
then this will be detected by the distinctness protocol. If the prover provided
quadruplets such that $\mathcal{E} \ne \mathcal{W}$ then, this will be detected by the $\SetEquality$
protocol. Thus, it must be the case that $\mathcal{R}=\mathcal{W}$ but the RAM program outputs 0.
This means that the out $y$ that $v^*$ receives is 1, and thus there must be an
inconsistency in the computation. However, by the canonical form of the
computation, we know that any inconsistency implies that $\mathcal{R} \ne \mathcal{W}$. 

\paragraph{Number of Rounds.}
The protocol $\pi'$ will have some addition rounds to $\pi$. In the first round, have the prover sends the IDs of the nodes. If $\pi$ begins with a message from the prover to the verifier, then this can be sent in parallel to this message. Otherwise, a simple solution is to have this as the first message before the message from the verifier. This would add the round complexity by 1. 

The last message of the protocol $\pi$ is an $\mathsf{M}$ message. At this point, we run a protocol to simulate the computation of $\verifier$. This is an $\mathsf{MAM}$ protocol. Thus, we can have the first message be in parallel to the last message of the $\pi$. Then, we have an addition 2 message.

Overall, if $\pi$ is an $r$ message protocol then $\pi'$ will be either $r+2$ message, if the prover goes first in $\pi$, or an $r+3$ message protocol otherwise.

We observe that we can avoid that addition first message in the case where $\pi$ starts with the verifier. In the first message, the nodes chooses random values $\alpha_u$ and send them to the prover. Then, we force the prover to send the IDs (in the second message) according to the ordering of the $\alpha_u$ values. Let $u_1,\ldots,u_n$ be the nodes ordered according to the $\alpha_u$ values and let $\ID_u$ be the ID given to node $u$. The prover additional sends each node $u_i$ the value $\alpha_{u_{i+1 \mod n}}$. Each node $u$ sets $a_u = (\alpha_u,\ID_u)$ and $b_u=(\alpha_{u_{i+1 \mod n}}, \ID_{u} + 1 \mod n)$. The nodes run a $\SetEquality$ equality protocol for $\A$ and $\B$ where $\A = \{a_u\}$ and $B=\{b_u\}$. We observe that the prover is honest if and only if $\A = \B$ (the reasoning is very similar to the soundness argument in the $\Distinctness$ protocol).

\section{Graph Asymmetry and GNI}\label{sec:asym}
The graph Asymmetry language consists of all graphs that do {\em not} have a non-trivial automorphism, i.e.\ every non-identity permutation of its vertices yields a different graph.
\begin{definition}[Graph Asymmetry]\label{def:problem-asymmetry}
	The language $\ASym$ contains all graphs that do {\em not} have a non-trivial automorphism (all asymmetric graphs).
\end{definition}

We show a public-coin protocol for graph asymmetry that uses our RAM program compiler. The protocol
consists of 4 rounds with $O(\log n)$ communication complexity. Formally, we
show that
\begin{theorem}
	$\ASym \in \dAMAM[O(\log n)]$.
\end{theorem}
We begin by a short description of a standard (centralized) interactive protocol for $\ASym$ which is a simple adaptation of the public-coin protocol for graph non-isomorphism \cite{GoldwasserS89,GoldreichMW91} (see also \cite{BabaiM88}). From hereon we denote this protocol by the ``$\GNI$ protocol''.

Let $S$ be the set of all graphs that are isomorphic to $G$. That is, $S=\{G':G \cong G'\}$. The main observation of the $\GNI$ protocol which follows here directly is that if $G$ has no (non-trivial) automorphism then $|S| = n!$ while if $G$ does have an automorphism then $|S| \le n!/2$. Thus, the focus of the protocol is on estimating the size of $S$.

The verifier samples a pair-wise hash function $g\colon \bit^{n^2} \to \bit^{\ell}$, where
$\ell \in \N$ is the smallest number such that $2^{\ell} \ge 2n!$ and sends it to the prover. The prover seeks for a graph $G' \in S$ such that $g(G')=0^{\ell}$. The main observation is that the probability that such a graph $G'$ exists is higher when $S$ is larger which will allows us to distinguish between the cases of $S$. As observed by \cite{KolOS18} using an almost-pairwise hash function suffices for these purposes, and a seed to such a function takes $O(n\log n)$ bits.
Note that to send $G'$ it is enough to send a permutation $\pi$ such that $G'=\pi(G)$, which can also be represented with $O(n \log n)$ bits. These facts are important, since if we wish to get an $O(\log n)$ distributed protocol per node, we must have the underlying  centralized protocol as communicating a total of $O(n\log n)$ bits.

Our goal is to simulate this protocol with a distributed verifier. The nodes have unique IDs (strings of length $O(\log n)$) and let $u_1,\ldots, u_n$ be the ordering of the nodes sorted by their ID. Let $i=I(u)$ be the index of node $u$ in this ordering, and let $\ID_i$ be the ID of node $u_i$. Note that a node $u$ does not know $i$ in advance, but the prover knows the entire ordering.\footnote{This is without loss of generality, as the nodes can pick unique IDs in the range $1$ to $n^3$ (uniqueness holds w.h.p.) and send them to the prover.}

As our first step, we hash the set $S$ to a set $S'$ of the same
size (w.h.p.), and where elements in $S'$ have a small representation. In particular,
while the graphs in $S$ are represented using roughly $n^2$ bits, the hashed
values in $S'$ will have only $O(n\log n)$ bits of representation. More
importantly, the hash used to compute $S'$ will be {\em local} and can be
easily computed by the nodes of network by considering only their own neighborhood.

In more detail, our hash function $h$ will be composed of $n$ hash functions.
Each node $u$ chooses a seed for an $\epsilon$-almost-pairwise 
hash function (according to 
\Cref{thm:almost-pairwise}):
$$
h_u\colon \bit^{n^2} \to \bit^{3\log n}
$$
where $\epsilon = 1/n$. The seed length is of size $3\log n + \log n + \log(1/\epsilon) = O(\log n)$ bits. Let
$h_1,\ldots,h_n$ be the $n$ chosen hash function ordered by the index of the nodes (\ie by $I(u)$). Let
$G=x_1,\ldots,x_n$ where $x_i \in \bit^n$ is the indicator vector for the
neighbors of node $i$ in $G$. Then, we define a hash function $h \colon \bit^{n^2} \to \bit^{3n\log n}$ as
$$
 h(G)=h_1(x_1) \concat \ldots \concat h_n(x_n).
$$
Using $h$ we can define the set $S' = \{h(G): G \in S\}$. It is easy to see that $|S'| \leq |S|$.
We note that the fact that $h$ is locally computable means that the probability 
of any two graphs colliding under $h$  is not as small as we would like (say in 
a pair-wise independent function) and we need finer analysis to show that  
$|S'|$ is close to $|S|$. The key  point is that $S$ contains graphs that are 
all isomorphic to each other and hence  we are able to show that with high 
probability it holds that $|S'| = |S|$ (this is shown in 
\Cref{claim:size-of-s}).

Assume that indeed $|S'|=|S|$. Then, we can continue to simulate the centralized protocol where we replace the set $S$
with the set $S'$. That is, to sample the
pairwise hash pairwise hash function $g_K\colon \bit^{n\log n} \to \bit^{\ell}$ which has a seed length of $O(n\log n)$ we let each node $u_i$ sample a chunk of the seed, $K_i$, of length $O(\log n)$. Then, we let the seed $K$ of $g$ to be $K=K_1,\ldots,K_n$ where again the ordering of the nodes $u_1,\ldots,u_n$ are according to their IDs. The prover knows this ordering and can construct $K$ accordingly.

The family of functions that we pick for the task is that of Corollary~\ref{cor:linear_ikos}: it has a succinct description and can be evaluated by a linear sized circuit over the field (or a linear-time RAM). It is crucial to use a hash function that can be computed in linear time since we are going to apply our RAM compiler on this computation and want a minimal overhead.

The prover sends the graph $G' \in S$ by sending the permutation $\pi$ in the following way: it sends the node $u$ its new ID $u'=\pi(u)$ in $G'$. Each node $u$ learns the IDs of their neighbors in $G'$, denoted by $N_{G'}(u)$. Moreover, the prover sends each node $u$ its index $i=I(u)$. The validity of these index will be checked next.
Our goal now is to verify that $g_{K}(h(G'))=0^{\ell}$, and to check the validity of the index's $i$ given by the prover. Let $y = h(G')$.
Notice
that $y$ can be computed locally, since $h$ is a local hash function, that is, each node computes $y_i$.
Then, we claim that the rest of the computation, \ie computing $g_{K}(y)$ can be performed by a linear-time RAM program and therefore applying our RAM program compiler of~\Cref{sec:RAM_compiler} on this linear time program will finish the last part of the protocol with proof of size $O(\log n)$.

That is, we write a RAM program $C$ as follows. The input to $C$ is a list of $n$ tuples $(i,\ID_i,K_i,y_i)$. The program composes the seeds $K_i$ according to the ordering to get the seed $K$. Then, it composes $y=y_1,\ldots,y_n$, computes $g_{K}(y)$ and verifies that $g_{K}(y)=0^{\ell}$. Finally, the program verifies the indexes $i$ given by the prover. It create an array $A$ of length $n$ and sets $A[i]$ to be the ID of the node with index $i$. Then, it traverses $A$ and verifies that $A[i] < A[i+1]$ which guarantees that the IDs where given by the right order.
This completes the description of the protocol, see \Cref{fig:protocol-asym} for more details.

\protocol
{A protocol for Graph Asymmetry.}
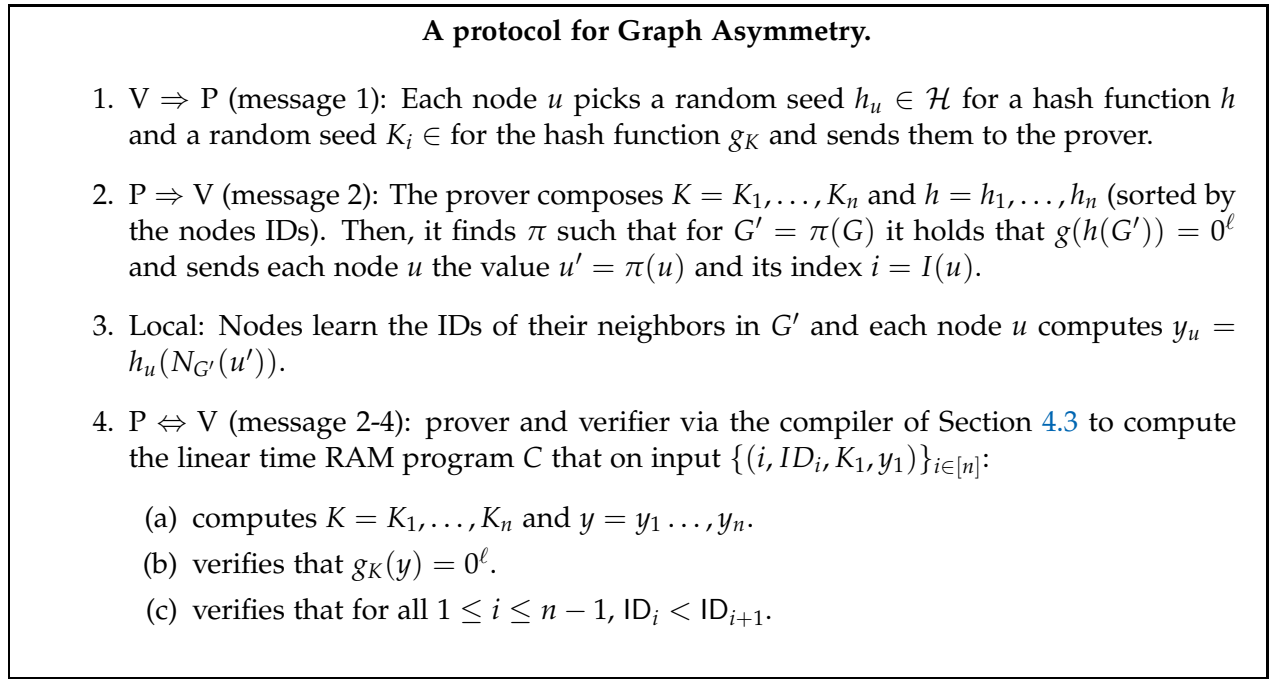
{A distributed AMAM protocol for graph asymmetry.}
{fig:protocol-asym}
{
\begin{enumerate}
	\item V $\Rightarrow$ P (message 1): Each node $u$ picks a random seed $h_u \in \cH$ for a hash function $h$ and a random seed $K_i \in $ for the hash function $g_K$ and sends them to the prover.
	\item P $\Rightarrow$ V (message 2): The prover composes $K=K_1,\ldots,K_n$ and $h=h_1,\ldots,h_n$ (sorted by the nodes IDs). Then, it finds $\pi$ such that for $G'=\pi(G)$ it holds that $g(h(G'))=0^{\ell}$ and sends each node $u$ the value $u'=\pi(u)$ and its index $i=I(u)$.
	\item Local: Nodes learn the IDs of their neighbors in $G'$ and each node $u$ computes $y_u=h_u(N_{G'}(u'))$.
	\item P $\Leftrightarrow$ V (message 2-4): prover and verifier via the compiler of \Cref{sec:RAM_compiler} to compute the linear time RAM program $C$ that on input $\{(i,ID_i,K_1,y_1)\}_{i \in [n]}$:
	\begin{enumerate}
		\item computes $K=K_1,\ldots,K_n$ and $y=y_1\ldots,y_n$.
		\item verifies that $g_{K}(y)=0^{\ell}$.
		\item verifies that for all $1 \le i \le n-1$, $\ID_{i} < \ID_{i+1}$.
	\end{enumerate}
\end{enumerate}
}

\paragraph{Completeness.}
To show completeness, we need to show that with high probability there is a
graph $G' \in S$ such that $g(h(G'))=0^{\ell}$. We show that the set
$S'$ (with high probability) will have the same size as $S$.
\begin{claim}\label{claim:size-of-s}
	$	\Pr_{h}[|S'| \ne  |S|] \le 1/n$.
\end{claim}
\begin{proof}
Let $H$ and $H'$ be two graphs, and define $\Delta(H,H')=k$ if they differ
on $k$ neighborhoods. That is, if
$H=x_1,\ldots,x_n$ and $H'=x'_1,\ldots,x'_n$ then
$$
\Delta(H,H') = |\{i : x_i \ne x'_i\}|.
$$
Let $I=\{i_1,\ldots,i_k\}$ be the set of $k$ indices on which $H$ and $H'$
differ. Then,
\begin{align*}
\Pr[h(H)=h(H')] = \Pr[\forall i \in I, h_{i}(x_i) = h_{i}(x'_i)]
\le \left(\frac{1+\epsilon}{2^{\ell}}\right)^{k} \le 1/n^{2k}.
\end{align*}
Thus, the probability of collision is smaller as $k$ is larger. We want to take a union over all pairs to show that there are no collisions. One concern here is the high collision probability for small values of $k$. However, what we show is that there are only a few graphs in $S$ that have small distance.

We bound the number of pairs of graphs in $S$ that have distance $k$. There
are
${n \choose k}$
possible locations $i_1,\ldots,i_k$ for which a pair $(H,H')$ of distance
$k$ can differ on. Fix a specific set $i_1,\ldots,i_k$. Since $H,H' \in S$, the
$k$ locations in each graph are a permutation of either $G$.
Thus, there are $k!$ possibilities for each graph, and $(k!)^2$ possibilities for the pair.
All together, we have
$$
|\{(H,H') \in S : \Delta(H,H')=k\}| \le {n \choose k}(k!)^2 < n^k k!.
$$
Thus, we can bound the probability that a colliding pair exists:
\begin{align*}
\Pr_{h}[\exists (H,H' \in S) : h(H)=h(H')] & \le
\sum_{k=1}^{n}
\sum_{\begin{subarray}{c}
	(H,H') \in S\\
	\Delta(H,H')=k
	\end{subarray}
}\Pr[h(H)=h(H')]
\le \sum_{k=1}^{n} \frac{n^k k!}{n^{2k}} < \frac{1}{n}~.
\end{align*}
\end{proof}

Therefore, we condition on the event that $|S|=|S'|$. Recall that $\ell$ is the smallest number such that $2^{\ell} \ge 2n!$.
Following the analysis of the $\GNI$ protocol we show that
\begin{align}\label{eq:1}
\frac{|S|}{2^{\ell}}(1 - \frac{|S|}{2 \cdot 2^{\ell}}) \le \Pr[\exists H \in S : h(H)=0^{\ell}] \le \frac{|S|}{2^{\ell}}~.
\end{align}
The upper bound follows by a simple union bound. For the lower bound we observe that
\begin{align*}
\Pr&[\exists H \in S : h(H)=0^{\ell}] \ge \sum_{H \in S}\Pr[h(H)=0^{\ell}] - \frac{1}{2} \cdot \sum_{H \ne H' \in S}\Pr[h(H)=0^{\ell} \wedge h(H')=0^{\ell}] \\
& \ge \frac{|S|}{2^{\ell}} - \frac{|S|^2}{2} \cdot 2^{-2\ell}
\ge \frac{|S|}{2^{\ell}}(1 - \frac{|S|}{2 \cdot 2^{\ell}})~.
\end{align*}
Let $p=2n!/2^{\ell}$ where $1/2 < p \le 1$. If $G$ has no automorphism then $|S|=n!$ and thus
$$
\Pr[\exists H \in S : h(H)=0^{\ell}] \ge \frac{n!}{2^{\ell}}(1-\frac{n!}{2\cdot2^{\ell}}) = \frac{p}{2}(1-\frac{p}{8}) =\alpha~.
$$
On the other hand, if $G$ has an automorphism then $|S| \le n!/2$ and thus
$$
\Pr[\exists H \in S : h(H)=0^{\ell}] \le \frac{n!}{2 \cdot 2^{\ell}} =\frac{p}{4} =\beta~.
$$
Thus, we have completeness $\alpha$ and soundness $\beta$. Since we can perform parallel repetition, it suffices to show that $\alpha - \beta$ is bounded by a constant. Since $p \ge 1/2$ it holds that
$$
\alpha - \beta = p/2(1-p/4) - p/4 = p/4-p^2/16 \ge 0.1
$$
Thus, repeating a constant number of times, we can push these parameters to $1-\epsilon$ and $\epsilon$ for any constant $\epsilon > 0$ while paying only a constant factor in the communication complexity. This completes the analysis of the protocol.

\subsection{Graph Non-Isomorphism}\label{sec:gni}
The protocol above for $\ASym$ can be easily adapted to solve the $\GNI$ problem with the same complexity. In the $\GNI$ the input is {\em two} graphs $G_0$ and $G_1$ however since there is only one network graph, there are several interpretations of what is the distributed analog of this problem. This is the reason we focused on the $\ASym$ problem where there is no ambiguity.
\begin{definition}[Graph Non-Isomorphism]
	The language $\GNI$ consists of all pairs of graphs $(G_0, G_1)$ where $G_0$ is {\em
	not} isomorphic to the graph $G_1$.
\end{definition}
Here we assume that the communication graph is the union of $G_0$ and $G_1$ where node gets a trinary inputs for each incident edge indicating if the edge is contained in $G_0$, $G_1$ or both.

When adapting the $\ASym$ protocol for $\GNI$ this results in a $\GNI$ problem where nodes can communicate on both graph $G_0$ and $G_1$. That is, there is one set of vertices $V$ for the network and both $G_0$ and $G_1$ are defined over $V$. The edge set of the network is the union of the edges of $G_0$ and $G_1$. The edges are marked if they belong to $G_0$ or $G_1$ or both.

The protocol we presented for $\ASym$ was an adaptation of the protocol for $\GNI$. In the $\GNI$ protocol, we define the set
$$
S=\{(G',\pi) : (G' \cong G_0 \vee G' \cong G_1) \text{ and } \pi \text { is an automorphism of } G' \}~.
$$
The key point is that if $G_0 \cong G_1$ then $|S| = n!$ while if $G_0 \not \cong G_1$ then $|S|=2n!$. Thus, the goal is to estimate the size of $S$ just as in the $\ASym$ protocol. One difference is that here we need to additionally verify that $\pi$ is an automorphism. We employ the $\dMAM[O(\log n)]$ protocol of \cite{KolOS18} for this. The result is the following
\begin{corollary}
	$\GNI \in \dAMAM[O(\log n)]$.
\end{corollary}
We note that while this improves upon the $\dAMAM[O(n\log n)]$ of \cite{KolOS18}, our
protocol works only when the $\GNI$ problem is defined such that nodes can
communicate on {\em both} graphs $G_0$ and $G_1$, where the protocol of
\cite{KolOS18} works also on the definition $\GNI$ where only $G_0$ is the
communication graph and $G_1$ is given as input nodes. In \Cref{sec:RRR} we show a
protocol for $\GNI$ in this harder definition as well, that has constant many rounds and
$O(\log n)$ bits.

\section{Compilers for Small Space and Low Depth Verifiers}\label{sec:RRR}
In this section  we present compilers that transform any centralized
prover-verifier interaction  on the graph where the verifier uses either small space or requires low-depth
into a distributed constant round
interactive protocol

We start with verifiers that require only small space and show that they can be be turned into a distributed constant round
interactive protocol with a small proof size. Formally, we show the following:
\begin{theorem}
There exists a constant $\delta$ such that if $L$ is a language that can be
decided in time $\poly(n)$ and space $S=n^{\delta}$ then $L \in
\dIP[O(1),O(\log n)]$.
\end{theorem}

The main tool behind this theorem is the interactive protocol of Reingold, Rothblum and Rothblum~\cite{ReingoldRR16}.
They show that for every language that can be evaluated in polynomial time and
bounded-polynomial space there exists a constant-round interactive proof
such that the verifier runs in (almost) linear time. This is an excellent
starting point for us, as our RAM compiler works great for compiling verifiers
that run in linear time.

However, a linear-time here is with respect to the size of the graph, \ie $m = O(n^2)$, and we wish to reduce the running time to $O(n)$ before we apply the compiler. As already observed in \cite{ReingoldRR16}, the running time of the verifier can be made sublinear (\eg $n^{\delta}$ for some small constant $\delta$) if the verifier is given oracle access to a low degree extension of the input (the input is the graph and possibly additional individual inputs held by each node). Luckily, computing a point in a low degree extension of the input is a task that is well suited for a distributed system, as it is a linear function of the input and hence can be computed ``up the tree''.

The following theorem is a simple adaptation of the result of Reingold et
al.~\cite{ReingoldRR16}. Here we state
the theorem with respect to inputs of length $m$ (the size of the graph), so as not to
confuse it with the parameter $n$, which in our context denotes the number of nodes
in the graph.
\begin{theorem}[Follows\footnote{In the original Theorem
		of \cite{ReingoldRR16} the number of queries to the low degree
		extension of the
		input is bounded only by $O(\poly(S) \cdot m^\delta)$. However, for low
		degree
		extensions, this can be reduced to a single query. The high level idea
		is to
		consider a low degree curve that agrees with all the queried points.
		The prover
		specifies the values for the points on the curve and the verifier
		checks answer
		on a random point on the curve. See \cite[Section 6]{KalaiR08} for
		further
		details.} from~\cite{ReingoldRR16}]\label{thm:RRR}
Let $L$ be a language that can be decided in time $\poly(m)$ and space $S = S(m)$, and $let \delta \in (0, 1)$ be an arbitrary (fixed) constant. There is a public-coin interactive proof for $L$ with perfect completeness and soundness error $1/2$. The number of rounds is $O(1)$. The communication complexity is $(poly(S) \cdot m^\delta)$. The (honest) prover runs in time $\poly(m)$, and the verifier runs in time $O(\poly(S) \cdot m^\delta)$, given a single query access to a low-degree extension of the input.
\end{theorem}

Let $L$ be a language that can be decided in time $\poly(m)$ and space $S=m^{\delta}$, where $\delta$ is a small enough constants such that the communication complexity and verifier running-time from Theorem \ref{thm:RRR} are $O(\poly(S) \cdot m^\delta) \le m^{1/2} \le n$. Thus, we can distribute the communication of the protocol between the nodes such that each node gets a single bit.

Then, we need to simulate the computation of the verifier on the transcript.
Since the running time of the verifier is $n$ (given oracle access to a low
degree extension of the input), we use the RAM program compiler of Theorem
\ref{thm:ram-compiler} to simulate this part. The number of rounds will grow
only by 2 and the proof size is $O(\log n)$.

Finally, we need to implement oracle access to the low degree extension of the
input. We explain exactly what this means and how to compute it. We give a
description of {\em low degree extensions}.

\paragraph{Low Degree Extensions.}
Let $\HH$ be a finite field and let $\F$ be an extension field of
$\HH$, such that $\HH \subset \F$. Fix an integer $m \in \N$, and let $\phi
\colon \HH^m \to \F$ be a function. It is well known that there exists a unique
extension of $\phi$ into a function $\widehat{\phi} \colon \F^m \to \F$ which
agrees with $\phi$ on $\HH^m$ such that $\widehat{\phi}$ is an $m$-variant
polynomial of individual degree at most $|\HH|-1$. The function
$\widehat{\phi}$ is called the low degree extension of $\phi$ with respect to
$\F$, $\HH$ and $m$.	

Furthermore, there exists a
collection of $|\HH|^m$ functions $\{\widehat{\tau}_x\}_{x \in \HH^m}$ such
that each $\widehat{\tau}_x \colon \F^m \to \F$ is an $m$-variant polynomial of
individual degree $|\HH|^m$, and for every function $\phi \colon \HH^m \to \HH$
it holds that
\begin{align*}
\widehat{\phi}(z_1,\ldots,z_m) = \sum_{x \in
\HH^m}\widehat{\tau}_x(z_1,\ldots,z_m) \cdot \phi(x)~.
\end{align*}

\paragraph{Oracle Access to the Low Degree Extension.}
Let $X$ be the input of the verifier. That is, $X$ contains the graph itself
and additional inputs that each node has (\eg randomness and arbitrary other
inputs). We interpret $X$ as
describing the truth table of a function $\phi \colon \HH^m \to \F$. Then, the
oracle of a low degree extension of the input is a query to the function
$\widehat{\phi}$.

Let $z=z_1,\ldots,z_m$ be the query performed by the verifier, where $z \in
\F^m$, and let $v \in \F$ be the expected point. That is, the task of the
verifier is to check that $\widehat{z}=v$. The point $z$ and $v$ is defined in
the
transcript of the protocol and thus each
node has a single bit of $z$. We let the prover broadcast $z$ and $v$ to all
the nodes
which in turn verify consistency with their local bit. Now, we need to compute
\begin{align*}
\widehat{\phi}(z_1,\ldots,z_m) = \sum_{x \in
	\HH^m}\widehat{\tau}_x(z_1,\ldots,z_m) \cdot \phi(x)~.
\end{align*}
where all the nodes know $z=z_1,\ldots,z_m$ and each node knows a part of the
truth table of $\phi(\cdot)$. In the \cite{ReingoldRR16} protocol, it was shown
how to compute this in (almost) linear time by a centralized prover. Here, we
show how to compute this by a distributed verifier (where local computation is
free).

Let $X_u$ be all the elements $x \in \HH^m$ such that the node $u$ knows
$\phi(x)$. This includes all edges incident to $u$ and all bits
of $u$'s additional input. Then, $u$ can locally compute $S_u$ where
$$
S_u=\sum_{x \in X_u}\widehat{\tau}_x(z_1,\ldots,z_m) \cdot \phi(x).
$$
Using this notation we
have that
$$
\widehat{\phi}(z) = \sum_{u \in G}S_u.
$$
Finally, the nodes compute the sum $\sum_{u \in G}S_u$ ``up the tree''.
That is, the prover sends and tree $T$ with root $r$ (along with a proof) and
for each node $u$ he sends the sum $\sum_{v \in T_u}S_v$, nodes check
consistency with
their children to assure these values. The root $r$ has value $S_r=\sum_{v \in
G}S_v = \widehat{\phi}(z)$ and verifies that indeed $\widehat{\phi}(z)=v$. This
completes the description of the compiled protocol.

\paragraph{Communication Complexity.}
In the \cite{ReingoldRR16} protocol, the field $\HH$ is set to be of size
$O(\log n)$ and the field $\F$ is of size $\polylog(n)$. The parameter $m$ is
set such that $m=\log_{|\HH|}(n) = O(\log{n}/\log \log n)$. Thus, we get that
the point $z \in \F^m$ can be written using $O(\log \log n) \cdot m = O(\log
n)$ bits. The variable $v$ can be written using $O(\log \log n)$ bits.
Altogether, the total communication received by a node is bounded by $O(\log
n)$. Finally, by repeating the protocol a constant number of times (in
parallel) we can improve the soundness arbitrarily  while increasing the proof
size by only a constant.

\paragraph{Using the compiler for $\GNI$.}
In \Cref{sec:gni} we have seen a protocol for the $\GNI$ problem. However, the protocol
works only for the setting in which both graphs $G_0$ and $G_1$ are communication
graphs. A more difficult formulation of the problem is where $G_0$ is the communication graph
and $G_1$ is given as input to the nodes. That is, each node gets as input its neighbors
in $G_1$ but it cannot communicate with them directly.

We observe that our compiler for small space can be used to get a protocol for the
$\GNI$ problem in this stricter formulation using a constant number of rounds and a
proof of size $O(\log n)$. For this, we need to show a standard (centralized) interactive
protocol (with public coins) where the verifier uses small space at the end to verify the interaction.

We observe that the standard Goldwasser-Sipser interactive protocol for
graph non-isomorphism, as discussed in Section~\ref{sec:asym},  can be implemented in small space by choosing the right hash
function.
We need a hash function $h \colon \bit^{n^2} \to \{n\log n\}$ that has a small collision
probability, and the new requirement is that  verifying that $h(G')=0$ can be done in small space. First, as in Section~\ref{sec:asym} we let $h$ be the concatenation of $n$ hash
functions $h_1,\ldots,h_{n/\log n}$ such that $h_i \colon \bit^{n^2} \to \bit^{\log^2 n}$.
The $h_i$'s are chosen independently by the nodes and sent to the prover. Each
$h_i$
is chosen from a family $\cH$ of almost pair-wise hash functions that can be evaluated in small
space (\ie $\polylog(n)$), and the seed length is $O(\log^2 n)$. Many
One example is to define
$$
h_{a,b}(x)= b + x_1a^1 + x_2a^2 + \dots + x_{\ell}a^{\ell} + a^{\ell+1},
$$
where  $x=x_1,\ldots,x_{\ell}$ and
$a,b,x_i \in GF[2^{\log^2 n}]$. The probability of collision under $h_{a,b}$ is bounded by $n^2/n^{\log n} \le 1/n^{10}$.
One can observe that this family has all the required properties. Thus, $h$ has seed length
$O(n\log n)$ and can be computed in $O(\log n)$ space. Using our compiler for small
space computations we get the required protocol.

\subsection{A Compiler for All $\NC$ Computation}

We have shown how to compile the \cite{ReingoldRR16} protocol into a distributed verification protocol
that has constant rounds and a proof of size $O(\log n)$. The crux of this solution is based on the fact that 
given oracle access to a low degree extension of the input, the verifier can be 
made very efficient. This allowed us to use the RAM compiler while supplying 
the low degree extension  via a computation on a spanning tree.
%

The protocol of Goldwasser, Kalai and
Rothblum~\cite{GoldwasserKR15} shares similar properties with the RRR protocol which will allows us to compile this protocol as well. The protocol of \cite{GoldwasserKR15} considers verifier whose computation can be implemented by low depth circuits (as opposed to small space).
Let the class ``uniform $\NC$'' be the class of all language computable by a
family of $O(\log(n))$-space uniform circuits of size $\poly(n)$ and depth
$\polylog(n)$.
They showed that for any
language computable by a uniform $\NC$ circuit there is a public-coin interactive
protocol where the verifier runs in time
$\polylog(n)$ given oracle access to a low degree extension of the input and the
communication complexity is $\polylog(n)$.

Similarly to our compilation of the \cite{ReingoldRR16} protocol, we can also compile the GKR protocol with a slightly larger cost. The number of rounds and proof size will be
$\polylog(n)$, compared to the $O(1)$-round and $\log n$ proof size in the case of \cite{ReingoldRR16}. We therefore have:
\begin{theorem}\label{thm:distributed-gkr}
For any language $L$ in uniform $\NC$, it holds that $L \in
\dIP[\polylog(n),\polylog(n)]$.
\end{theorem}
 Actually, the formal statement is more general. If $L$ can be decided by circuits of depth
$d$ and size $T$ then $L \in \dIP[d \cdot \log T \cdot \polylog(n), d \cdot \log T \cdot
\polylog(n)]$ which is non-trivial even for circuits of depth $n$.

One type of problems for which the GKR based protocol may be more useful than the
RRR based one is for problems based on shortest path problems. For instance, proving
that the diameter of the graph is a given value. This problem is in $\NC$ and 
hence our protocol is applicable to it.  
\section{Below the $\log n$ Barrier}\label{sec:loglog}
Our protocols so far used $O(\log n)$ sized proofs, which appears as  a very natural
barrier as simple tasks as counting and even pointing to a neighbor seem to
require $\log n$ bits. Nevertheless, in this section we show to ``push'' the protocols above
to use only $\log \log n$ at the price of more interaction rounds. Our main
result is a 5-message protocol for the {\em Decisional} Symmetry problem which is the same as $\Sym$ except that the permutation is fixed and the task is to decide whether it is an automorphism. Kol \etalcite{KolOS18} showed that $\DSym \in \dAM[O(\log n)]$.

We show that by adding more rounds of interaction we can further reduce the
proof size to $O(\log \log n)$.
\begin{theorem}
	$\DSym \in \dMAMAM[O(\log \log n)]$.
\end{theorem}
We begin by presenting a simple $\dMAM[O(\log)]$ protocol, then we show how to
simulate this protocol using $O(\log \log n)$ bits. Let $G$ be the communication graph and let $G'=\pi(G)$ be the graph after applying the fixed permutation. Each node $u$ knows its neighbors $N(u)$ in $G$ and its neighbors $N_{G'}(u')$ in $G'$ where $u'=\pi(u)$. Then check that the two graphs are equal, we need to verify that the set of edges in $G$ and in $G'$ are the same. Thus, we run the $\SetEquality$ protocol on the two sets of edges.

We re-develop the basic ``distributed NP'' tools but pushed down to the
$O(\log \log n)$ regime. Similar to the generality of the basic tree
construction in distributed NP proofs, these tools are basic and can be used
for many other problems as well. We show how to compute a tree in the graph
using only $O(\log \log
n)$ bits.

\subsection{Tool 1: Constructing a Rooted Tree}
The basic tool used for all our protocol was a spanning tree of a graph. Moreover, a
useful property of a tree is that it defines a unique node in the graph, the root,
which plays an important role in the protocols above.
While constructing this tree is simple using messages of length $\log n$ bits, it is
a challenging task using a proof of only $O(1)$ bits.

In the first message of the protocol, the prover computes a BFS tree in $G$
rooted at an arbitrary node $r$. Here it is crucial that the tree is a {\em
BFS} tree.
Then, it
sends each node $u$ its distance from the root modulo 3, denoted by $d_3(u)$.
Nodes exchange the value $d_3(u)$, to learn the distances of their neighbors.
Recall that in a BFS tree, the neighbors of a node $u$ can be either in the
same level of the
tree, one level higher, or level lower. The $d_3(u)$ values enable each node
to partition its neighbors into these three groups: a neighbor $v$ such that
$d_3(v)=d_3(u)$ is in the same level as $u$, if $d_3(v) = d_3(u)-1 \mod 3$ then
$v$ is one level higher than $u$ and if $d_3(v) = d_3(u)+1 \mod 3$ then $v$ is
one level lower than $u$. Each node $u$ sets its parent in the tree
$\parent(u)$ to be
its neighbor $v$ with $d_3(v) = d_3(u)-1 \mod 3$ with the minimal port number.
If no such neighbor exists, then this node is a root.

Let $T$ be the resulting graph. That is, $T$ is defined on the vertex set
and has edges $\{(u,\parent(u))\}_{u \in G}$.
Note that if the prover is honest, then the graph $T$ is indeed a tree,
however, it might be different than the BFS tree computed by the prover. In any
case, the only property we require from $T$ is that it  be a spanning tree of $G$. If the
prover
is dishonest, then $T$ might not be a tree at all, and in particular might contain
cycles.

To combat such cycles,
each node samples a uniform bit $b_u$ and sends it to the prover. For each node $u$ let
$P_u$ be the path from the node $u$ to the root $r$ in the (alleged) tree $T$. The prover
sends each
node $u$ the value $s(u)=\sum_{v \in P_u}b_v \mod 2$, that is the sum of the $b_v$'s
on the path from $u$ to the root modulo 2. Nodes exchange these values with their parent
in the tree. Each node $u$
verifies that $s(u)=s(\parent(u)) + b_u \mod 2$. If $T$ contains a cycle, then we claim that with
probability at least $1/2$ the nodes will reject. Indeed, let $C$ be a cycle in $T$. If $\sum_{u \in
C}b_u= 1
\mod 2$ (which happens with probability $1/2$), then
the values $s(u)$ on this cycle must be inconsistent and thus there will be at least one node that
rejects.

So we know that $T$ contains no cycles or the cheating prover is caught with a reasonable probability. However, it might still be the case that $T$
is a forest. In such a case it will contain more than one root node. To eliminate this, we
have
the prover broadcast the value $b_r$ to all nodes in the network, which in turn check for
consistency. If there are more than one roots in $T$, then with probability $1/2$
their $b_r$ values will be different and thus nodes will detect this inconsistency.
This insures that $T$ has no cycles and a single root thus it must be a spanning tree of
$G$. Of course, the soundness can be
amplified by standard (parallel) repetition.

The result of computing a tree is that there is a single root $r$, a unique chosen node
among the nodes in the network. In particualr, the result is a protocol for
``leader election'' that uses a small proof.
\begin{corollary}
	$\LeaderElection \in \dMAM[O(1)]$.
\end{corollary}

The formal protocol is given in \Cref{fig:protocol-tree}.

\protocol
{A protocol for computing a tree $T$ in $G$.}
{A distributed protocol for computing a tree in the graph with $O(1)$ bits of proof.}
{fig:protocol-tree}
{
\begin{enumerate}
	\item $\prover \Rightarrow \verifier$ (message 1): prover picks an
	arbitrary root node
	$r$ in the graph,  computes a BFS tree $T$ rooted at $r$, and sends
	each node $u$ the value $d_3(u)=\depth(T,u) \mod 3$, \ie that depth of
	the node $u$ in the tree $T$ mod 3.
	\item Local: nodes exchange the value $d_3(u)$ and each node $u$ sets
	$\parent(u)$ to be its neighbor $v$ with $d_3(v) = (d_3(u)-1 \mod 3)$ breaking ties using
	port number (e.g., preferring the smaller port number). If no such node $v$ exists, then $u$ is a root. Nodes
	notify their parents so that parents learn all their children in the tree.
	\item $\verifier \Rightarrow \prover$ (message 2): each node $u$ samples a random
	bit $b_u$ and sends it to the prover.
	\item $\prover \Rightarrow \verifier$ (message 3): prover sends node $u$
	the values $s(u) = \sum_{v \in P_u}b_v \mod 2$ ($P_u$ is the path from $r$
	to $u$ on the tree $T$) and $b_r$.
	\item Local: nodes exchange values $s(u)$ and $b_r$. They verify that
	$s(u)=s(\parent(u))+b_u \mod 2$ and that $b_r$ is the same among their
	neighbors.
\end{enumerate}
}

\paragraph{Completeness.} Completeness follows directly from the construction.
The honest prover picks exactly one root $r \in S$. Computes a BFS tree $T$
rooted at $r$,  and gives the correct distances modulus 3 and the correct values
$s(u)$. Thus, all the
consistency checks of nodes will pass.

\paragraph{Soundness.}
First, observe that regardless of the values sent by the prover, each node
(except the roots) will identify a single neighbor as its parent in the
graph $T$.

\paragraph{Case 1: The graph $T$ has no root.}
If $T$ has no root then it must contain a cycle. Let $C=v_1,\ldots, v_k$ be such a cycle.
With probability half it holds that $\sum_{v \in C}b_v \mod 2=1$. Recall that each
node verifies that $s(u)=s(\parent(u))+b_u \mod 2$. Thus, we have that $s(v_i) = s(v_{i+1})
+ b_{v_i} \mod 2$. However, since $\sum_{v \in C}b_v \mod 2=1$ there are no set of
values $s(v_i)$ that will satisfy these conditions.

\paragraph{Case 2: The graph $T$ has more than one root.}
Let $r,r'$ be two roots. Then, with probability half it holds that $b_r \ne b_{r'}$. In such a
case, the prover cannot broadcast value 0 or 1 and there must be a node that will reject
this consistency check of the broadcast.

\subsection{Tool 2: Proofs that Grow with the
Degree.}\label{tool:degree-sized-proofs}
We have constructed a tree $T$ in the graph $G$. The degree of a node $u$ in
the tree $T$ is the number of children $u$ has in $T$, denoted by
$\Delta(u)=\Delta_T(u)$.
In the rest of the protocol, it would be very helpful if each node $u$ could
get a proof of size $O(\Delta(u) \cdot \log \log n)$. However, some nodes might
have large $\Delta(u)$ and we cannot send such a proof directly.

Instead, we let each node $u$ get the proof of its parent $\parent(u)$. The
leaves of the tree get their own proof in addition to the proof of their
parents. Then, each node sends its proof to its parent. Since each node has at
most one parent in the tree, its gets a single proof of size $O(\log \log n)$
(the
only exception is the leave which get two such proofs). If a node $u$ has
degree $\Delta(u)$ then it gets $\Delta(u)$ proofs from its children, each of
size $O(\log \log n)$. Thus, we can simulate each node $u$ receiving a proof of
size $O(\Delta(u)\log \log n)$ within our $O(\log \log n)$ budget. From hereon,
we will describe the protocol with such a proof size and at the end such a
transformation is applied.

\subsection{Tool 3: Decomposition into Blocks}\label{sec:block-decomposition}
We have constructed a tree in the graph and have also increased the proof
capacity of nodes with high degree. We air to further increase the proof
capacity. The high level idea is to decompose the tree $T$ connected components
of
size roughly $\log n$. Then, each block can {\em act as a super-node} that
has capacity $\log n$ even if each real node gets only a single bit. Then, we
need
to simulate what the single node would have computed in a distributed manner in
the block.

We devise a protocol to decompose the tree $T$ into {\em
edge-disjoint} subtrees $T_1,\ldots, T_k$, which we call \emph{blocks}. Let $H$
be a graph with $k$ vertices, where node $v_i \in H$ corresponds to a block
$T_i$. There is an edge $(v_i,v_j)$ if $T_i \cap T_j = \{r\}$ where $r$ is a
root of $T_i$ and not of $T_j$. We
require the following from the decomposition protocol:

\begin{enumerate}
	\item $\forall i \in [k-1]$ it holds that $|T_i| \in [\log n, 3\log n]$.
	\item Blocks intersect at roots: If $i \ne j$ then $|T_i
	\cap T_j| \le 1$, and if $w \in T_i \cap T_j$ then $w$ is a root.
	\item The graph $H$ is a tree, and if $T_j$ is the parent of $T_i$ in $H$
	and $r$ is the root of $T_i$ then $\parent(r) \in T_j$.
	\item Each node $u$ knows its neighbors inside each block it belongs to.
\end{enumerate}
Computing such a decomposition by a centralized algorithm is simple and can be
done ``from bottom up'' on the tree, while greedily packing nodes into blocks.
Thus, we first let the prover compute this decomposition which we describe
next. Then, we describe how the prover sends the result back and proves that he
indeed computed the decomposition correctly. The whole decomposition is
performed in the first round of the protocol.

\paragraph{The Centralized Algorithm.}
The decomposition is
computed greedily starting from the leaves and working level by level up to the
root. In level $i$, if there is a node
$u$ such that the size of the subtree $T_u$ is in $[\log n, 2\log n]$ then
declare $T_u$ as a block and remove all nodes in $T_u$ except for $u$.

If there
is a node $u$ such that the size of the subtree $T_u$ is greater than $2\log n$
then we traverse its children $v_1,\ldots, v_\Delta$ according to the port
ordering and we greedily pack them into blocks each of sizes in $[\log n,2\log
n]$ where each block has $u$ as its root (this can be done since $|T_{v_j}| <
\log n$ for all $j \in
[\Delta]$). For each block, we remove all nodes of the block except the root
$u$.

We continues this for all
levels of the tree $T$. At the top level, there might be at most $\log n$
remaining nodes, and we add them to the last declared block. This is the only
block that
might have size more than $2\log n$ (but at most $3\log n$). This
completes the description of the algorithm.

One can easily verify that this decomposition satisfies properties 1-3 (where
the last property is described next). The blocks are by definition of size at
least $\log n$ and at most $3\log n$. The only intersection between blocks is
the roots that are not removed when a block is declared. The edges in $H$ are
always between a root $r$ to a block that is at a higher level in the tree, and
thus $H$ is a tree.

\paragraph{How to Distribute the Output.}
We distinguish between three types of nodes: (1) those who are not a root any
block, (2) those who are a root in exactly one block and (3) those who are
roots in more than one block. Thus, we let the prover send a trinary value to
each node indicating each type.

For nodes who are not roots, the output is very
simple. These nodes must be a member of only one block and their neighbors in
this
block are their neighbors in $T$. Thus, no additional information is required
from the prover.

Nodes that are root in a single block, are a
leaf in the other block. Thus, their children in $T$ are their neighbors in the
first block, and their parent is their only neighbor in the second block.
Again, no additional information in needed from the prover for these nodes to
know their neighbors in the block.

The third type is a bit more complicated. A
node that is a root in many blocks means that its children have been greedily
grouped in several blocks. However, since the algorithm groups according to the
port ordering of the root node, it suffice to know the {\em first} node in each
block in order for the root to divide its children according to the blocks.
Thus, we have the prover mark the first node of each block and these nodes
notify their parent \ie the root, that they are first in their block. This way,
the root knows exactly which children are
in each block.

\paragraph{Soundness.}
First, observe that no matter what a dishonest prover sends, by the definition
of the output of the protocol, the nodes will be decomposed into edge-disjoint
blocks that satisfy properties 2-4.

The only thing we need to verify is that the size of each block is indeed in
the desired range of $[\log n, 3\log n]$. The prover sends each block $T_i$ a
proof of the
number of nodes in $T_i$. That is, each node $u$ in $T_i$ gets the size of its
subtree inside the block $T_i$. Nodes check consistency with their neighbors in
$T_i$. Since that blocks are supposed to be of size $O(\log n)$, the partial
sums can be described using $O(\log \log n)$ bits.

 We note that in this protocol, a root will have a proof of size that is
proportional to the number of blocks that is participates in (which is bounded
by its degree). However, as described before, this can be reduced using the
transformation described above (see \Cref{tool:degree-sized-proofs}).

\subsection{Tool 4: Set Equality via Super Protocols}
Our protocol has computed thus far a spanning tree $T$ in the graph $G$, a
decomposition of $T$ into blocks $T_i$ of size $O(\log n)$ and a super tree $H$
between the blocks. The
advantage of having such blocks is that their {\em total} proof capacity is
$\log n$.
That is, if we consider the ``super'' graph $H$, we can send each node in $H$ a
proof of size $O(\log n)$ by sending each node inside the block $O(1)$ bits.
Then we can run protocols on the graph $H$ with very small proof size. These
would be call ``super protocols'', and would let us simulate the $\SetEquality$
protocol from \Cref{sec:set-equality} using only $O(\log \log n)$ bits.

In the $\SetEquality$ protocol, the root of the tree $T$ chooses a random field
element $s$ which it sends to the prover. The prover broadcast $s$ to all nodes
and they verify that they all got the same value. The element $s$ is $O(\log
n)$ bits long. Thus, we simulate this oh $H$ by having the whole block of the
root choose $s$ together. Then, in the original $\SetEquality$ protocol, then
prover sends each node $u$ the value $A_u=\prod_{v \in T_u, i \in
[\ell]}(a_{v,i} - s)$ (and
similarly for $B$. The prover here will send $A_u$ for each node of $H$ to the
corresponding block.

Now, our goal is to verify that the prover gave correct proofs. That is,
every two neighboring nodes in $H$, need to verify that the have received the
same element $s$, every node $u$ need to compute $A_u$ with its children in
$H$. We show how to simulate simple'' one-round protocols in
$H$ with $O(\log n)$ proof size using only $O(\log \log n)$ bits in the tree
$T$.
Here the term ``simple'' refers to protocols where the local computation
performed by each node is an aggregate function (see
\Cref{def:aggregate-function}). Aggregate functions are functions that works on
many inputs but that can be computed by applying the same operation on pairs of
inputs. The main examples we use are ``Equality''
(verifying that all neighbors have the same value in the proof) and ``Field
Op''
(\eg computing the sum/product of all neighbors value over a field). If $f$ is
an aggregate function then there exists a function $g$ such that in order to
compute $f(x_1,\ldots, x_n)$ it suffices to compute $y_i =
g(x_i,y_{i-1})$ for all $i$ where $y_1=x_1$ (if the function is Equality,
then $g$ would simply be the equality of its inputs, and in case of a ``Op''
then $g$ would be an addition/multiplication field operation).

To simulate the protocol we distributed the $O(\log n)$ proof for a node among
the $O(\log n)$ nodes in the corresponding block, such that each node gets
$O(1)$ bits of the proof, and a $\log \log n$-long index $i$ indicting the
location of these bits in the proof. Let $x_i$ be the proof of block $T_i$.
Then, we need to simulate the local computation, $f$, of a node in $H$ with its
children in the tree $H$. We let the prover compute $y_i = g(x_i,y_{i-1})$ and
since each
$y_i$ is $O(\log n)$ bits long the prover distributes $y_i$ to block $i$. What
is left is to
{\em verify} that indeed the $y_i$ values given by the prover are correct. That
is,
for each $i$ we need to verify that $y_i = g(x_i,y_{i-1})$, where block $i$ has
$y_i$ and $x_i$ and block $i-1$ has $x_{i-1}$ and blocks $T_{i-1}$ and $T_i$ a
shared parent in $H$.

Consider the block $T_i$ and its $k_i$ children $T_{i,1},\ldots, T_{i,k_i}$. Algorithm $\A_i$ consists of $k_i+1$ sub-algorithms that will run in parallel.
The role of Algorithm $\A_i$ is to verify that $y_i = g(x_i,y_{i-1})$. The main
idea is to define a graph $G_{i}$ which contains $T_i$, $T_{i+1}$ and a path
that connects them. Then,
since this graph has $O(\log n)$ nodes, we can run our RAM compiler on $G_{i}$
to compute $y_i = g(x_i,y_{i-1})$. Since the centralized algorithm can be
computed in time $O(\log n)$, the cost of such a protocol would be $O(\log \log
n)$.

The idea above works nicely if $T_i$ and $T_{i+1}$ share a root. Then, all
nodes know their neighbors in $G_i$ and can run the compiler protocol.
Moreover, we can run all such protocols in parallel. Each block participates in
at most two-such protocols (one with the previous block, and one with the
successor block). A root $r$ of a block might participate in $\Delta(r)$ such
protocols, which means that it will get $\Delta(r)$ proofs. Again, using
\Cref{tool:degree-sized-proofs} this can be reduced to $O(\log \log n)$ bits.
After this phase, we are left with all $i$ such that $T_i$ and $T_{i+1}$ do not
share a common root.

However, since the parent block is of size $O(\log n)$ and the remaining
children are vertex disjoint, then there can be at most $\log n$ such trees.
Thus, let $G'$ be the graph
containing the parent tree and all the left children. Then this graph has size
at most $n'=\log^2 n$ and we can again run the RAM program compiler on it to
compute all the values $y_i$. The cost of this will be a proof of size $O(\log
n') = O(\log \log n)$.

The final result is a $\SetEquality$ protocol. Notice that the operations in
the original $\SetEquality$ protocol are actually aggregative functions.
We use a ``super protocol'' to
verify that all blocks received the same element $s$. Then, we use another
super protocol to compute the product each the children of a block to compare
with the value $A_u$. Finally, the root block will hold $A_r$ and $B_r$ and we
can run a RAM compiler inside the block to verify their equality.

\paragraph{Rounds.}
The tree is computes in message 1 and verifier in messages 2-3. In message 2
the root block chooses $s$ and sends it to the prover. In message 3 the prover
responds with $s$ and with $A_u$ and $B_u$ for each block. Then, we run the
different super protocols which take 3 rounds as to run the RAM compiler which
will be sent as messages 3-5 (message 3 is used for $s,A_u,B_u$ and for the
first message of the super protocols). In total this is a $\mathsf{MAMAM}$
protocol that uses $O(\log \log n)$ bits.

\begin{corollary}
	$\SetEquality \in \dMAMAM[O(\log \log n)]$.
\end{corollary}

\protocol
{A protocol for computing a tree $\SetEquality$.}
{A distributed protocol for computing $\SetEquality$ in the graph with $\log
\log n$ bits of proof.}
{fig:protocol-set-equality-loglog}
{
\begin{enumerate}
	\item $\prover \Leftrightarrow \verifier$ (message 1-3): prover and
	verifier interact to compute a tree $T$ and a block decomposition of the
	graph.
	\item $\verifier \Rightarrow \prover$ (message 2): The root block $T_r$
	distributively samples an element $s$.
	\item $\prover \Rightarrow \verifier$ (message 3): prover sends $s$ to all
	blocks, such that each node gets $O(1)$ bits.
	\item $\prover \Rightarrow \verifier$ (message 3): prover sends each block
	the value $A_u$ and $B_u$.
	\item $\prover \Leftrightarrow \verifier$ (message 3-5): prover and
	verifier run a super protocol for equality of $s$ in all blocks.
	\item $\prover \Leftrightarrow \verifier$ (message 3-5): prover and
	verifier run a super protocol for ``product'' to verify the values $A_u$
	and $B_u$.
\end{enumerate}
}

\subsection{$\DSym$, $\Clique$ and More}
The tools described above are quite powerful in the sense that they allow us to solve several different problems using a proof of size $O(\log \log n)$, except for $\SetEquality$. One particular example is the problem $\DSym$ which is similar to the $\Sym$ problem except that the automorphism is fixed. That is, a permutation $\pi$ is given to all nodes, and the goal is to decide if $\pi$ is an automorphism of the graph. This problem was studied by \cite{KolOS18} where they showed that $\DSym \in \dAM[O(\log n)]$ but any distributed $\NP$ proof for requires a proof of size $\Omega(n^2)$.

We show that using more interaction, we can reduce the proof size to $O(\log \log n)$. Each node knows its neighbors in the graph $G$, and applies $\pi$ to learn its neighbors in $\pi(G)$. Now, we need to verify that the set of edges in $G$ is the same as in $\pi(G)$ which is simply solved by a running the $\SetEquality$ protocol described above. Thus, we get the following corollary:
\begin{corollary}
	$\DSym \in \dMAMAM[O(\log \log n)]$.
\end{corollary}

\paragraph{Summing Up the Tree.}
Another application of the super protocols described above is that we ``sum up the tree'' within the $\log \log n$ capacity. Suppose each node has a value $a_i \in \bit^{O(\log n)}$ and we wish to verify that $\sum_{i=1}^{n}a_i=K$ where $K$ is known to all.

For every root $r_i$ of block $T_i$ the prover computes $A_{r_i} = \sum_{u \in T_{r_i}}a_u$ the sum of values in the subtree rooted at $r_i$. Note that $A_{r_i}$ is $O(\log n)$ bits long and the prover cannot send it to $r_i$. Instead, the prover distributed $A_{r_i}$ among the nodes of the block $T_i$.

Then, we run a super protocol for the ``addition'' operation. That is, here $y_i = x_i + y_{i-1}$. This ensures that for the root $r$ of $T$ is in the block $T_r$ then this block has the correct value $\sum_{i=1}^{n}a_i$. Finally, to check if this value equals $K$ we run the RAM compiler inside the block $T_r$. The total proof size of this protocol is $O(\log \log n)$. Moreover, this can be done in three rounds: the tree and block decomposition are given in the first message and verified in messages 2-3. The values $A_{r_i}$ are given also in the first message. Then, we run the RAM compiler which is three messages that can be perform in parallel to messages 1-3.

Using this we get several different problems in the $O(\log \log n)$. For
example, it was shown by \cite{KormanKP10} that the problem $\Clique$ of
proving that a graph contains a clique of size $K$ (where $K$ is a fixed
parameter) can be done by a distributed NP proof of size $O(\log n)$. Plugging in our
addition operation we get a 3-round protocol with $O(\log \log n)$ proof size.

For the particular problem of $\Clique$ we get actually modify the protocol to depend
only on the leader election protocol and get a constant size proof. The prover marks a
clique of size
$K$ and selects one of the nodes in the clique to be a leader. We run the leader protocol
described above to verify that indeed a single leader is selected. Finally, each marked
nodes verify that indeed $K-1$ of its neighbors are marked {\em and} that one of them is
the
leader. If each node has $K-1$ neighbors then we know that there are at least $K$ nodes
marked as the clique. If there are more, then there will a node that is not the neighbor of
the selected leader. Thus,  this assures that there are exactly $K$ marked nodes and that
they form a clique. Formally, we ge that
\begin{corollary}
	$\Clique \in \dMAM[O(1)]$.
\end{corollary}

\section{Extension and Open Problems}\label{sec:conclusions}

We have shown that a  distributed verifier interacting with a prover in a randomized manner is very powerful. To a large extent our results show that it will hard to prove lower bounds in this model, especially super-polylogarithmic lower bounds.

\subsection{Argument Labeling Systems}
Can the interaction be eliminated? As discussed in Section~\ref{sec:our-results} this is not possible without changing the model. A common approach for eliminating interaction is the Fiat-Shamir transformation or heuristic (first used in~\cite{FiatS86}) that converts a public-coins interaction into one without interaction. In the Fiat-Shamir setting the parties have access to a random oracle, and the prover is computationally limited: it can only perform a (polynomially) bounded number of queries to the random oracle.
This results in an {\em argument} system rather than with a proof system. In such a system, proofs of false statements {\em exist}, but it is computationally hard to find them. Therefore, such protocols do not contradict the lower bounds for proof labeling schemes. We call such a system an ``argument labeling scheme''.

Applying such a transformation should be done with some care in general, and even more so in our setting. First, the error probability should be small, say, $2^{-\lambda}$ where $\lambda$ is what is known as the ``security parameter''. That is, the cheating prover has limited running time, but gets $1^{\lambda}$ (that is, $\lambda$ in unary representation) as input. We need the running time of the prover as a function of $\lambda$ to be significantly less than $2^{\lambda}$. There are general statements regarding the type of protocol for which the Fiat-Shamir transformation preserves soundness works (see~\cite{CanettiCHLRR18} and references therein). These include constant round protocols and the GKR protocol, so the protocols considered in this work are covered.

To use the Fiat-Shamir transformation in the distributed setting, we need to
apply the random oracle $R$ to the {\em entire input}, in our case, the graph.
While each node has access to the random oracle, they still do not know the entire
graph and thus cannot compute $R(G)$. Instead, we let each node apply $R$ to
its local neighborhood. Then, we combine all the results using a spanning tree.
That is, for a node $u$, let $v_1,\ldots,v_k$ be its children in the tree.
Suppose that $R$ compresses strings of arbitrary length into strings of length $\lambda$.
We define the values $y_u$  computed by node $u$. The computation works from the leaf nodes up to the root. The leaf nodes set the value
$y_u=R(N(u))$ (where $N(u)$ are the neighbors of $u$). Then, an inner node $u$
where children's values $y_1, y_2, \ldots, y_k$ computes $y_u=
R(y_1,\ldots,y_k,N(u))$. Finally, the value of the oracle is $y_r$ where $r$ is the root vertex
(this is also known as ``Merkle Tree computation").

Note that since the
 tree is chosen by the prover, it introduces another source of cheating for a dishonest prover.
For each tree used by the prover,it gets a single value
 of the oracle function. It can be proven that every change in the tree
 requires at least one new call to the oracle (w.h.p.). 
Since the prover has limited running time, it allows him to try a bounded number of trees, in fact, polynomially many, which can be shown to have an only negligible
 effect on the success probability.

With this approach, we can obtain an argument labeling systems with $O(\lambda \log n)$ bits for all the problems discussed in this work in a setting where the prover is polynomial given the witness. This includes ``permutation" and ``Symmetry" as well as all problems solvable small space and in NC.

However, once we settled for computational soundness we can get even more general results: we can apply it to the setting of verification of computation in the style of Kilian~\cite{Kilian92} and Micali~\cite{Micali00} (see also Barak and Goldreich~\cite{BarakG08}) where the correctness of a computation inside $\PP$ is proven using a PCP proof.
The proof itself is not communicated to the verifier but rather committed using
a collision resistant hash function (a cryptographic primitive which  exists 
in the random oracle world (as well as from collision resistant hashing).

The main challenge in incorporating these techniques in our setting is that the
input is assumed to be encoded by a linear code, as in \cite{BabaiFLS91}. In
our case, the input is the graph and we cannot encode the graph using a
distributed verifier since if it is dense it is too large. Instead, we have the
prover encode the graph and broadcasts  a short commitment of the encoding to
the nodes. Then, the verifier needs to check a small number of locations in the
PCP proof (say, $\log n$ locations) and a similar amount in the encoded input.
The network needs to verify the PCP proof of \cite{BabaiFLS91}. Some of the
work can simply be done
by a chosen leader, but the sensitive part that requires ``the whole village"
is verifying the correctness of these locations in the encoded graph. This {\em
can} be done by a distributed verifier: If $A$ is the generating matrix of the
code ($A$ is a fixed matrix known to all), then computing a single location in
the encoding corresponds to the inner product of $v$ and $A_i$ where $v$ is a
$n^2$ long vector that represents the graph, $i$ is the index in the encoded
word and $A_i$ is the $\ith{i}$ row of $A$. Each node can compute locally the
its neighborhood the inner product with the corresponding locations in $A_i$
and the full inner product can be easily computed ``up the tree'' in the graph
as we have seen. So the result is that {\em any problem in $\PP$ has an
argument labeling scheme of length $\lambda \cdot \polylog n$ in the random
oracle model}.

An intriguing question is whether the recent exciting results on using the Fiat-Shamir method without random oracles are relevant in the distributed setting (e.g.~\cite{KalaiPY18,CanettiCHLRR18}). A reasonable modeling assumption in this setting is the common random string or common reference string.

\subsection{Open Questions}

There are many interesting questions arising from this work (see also the questions in the body of the paper).
Can we tradeoff interaction for communication? We have seen example where going 
from $\dAM$ to $\dMAM$ (Symmetry) reduce the communication exponentially, so at 
the very least we expect to pay a significant cost. This is particularly useful 
when the   communication is $O(1)$.  Is there a general reduction from private 
coins to public coins with a distributed verifier? Does having {\em shared} 
private randomness help?

Finally, a natural property to consider is distributed protocols that are 
zero-knowledge. For public-coins protocol we 
note that if our compiler (any one of them) gets a zero-knowledge protocol as 
input then the output protocol will also be zero-knowledge. One question is whether we can do 
more in this model.

\section*{Acknowledgments}
We are grateful to Guy N.\ Rothblum and for in depth discussion about
interactive protocols and in particular for the work on
Theorem~\ref{thm:distributed-gkr} which is joint with him. We also thank Ron
D.\ Rothblum for explaining \cite{ReingoldRR16} and the parameters in it.

\bibliographystyle{alpha}
\bibliography{crypto}

\newcommand{\etalchar}[1]{$^{#1}$}
\begin{thebibliography}{CCH{\etalchar{+}}18}

\bibitem[AKY97]{AfekKY97}
Yehuda Afek, Shay Kutten, and Moti Yung.
\newblock The local detection paradigm and its application to
  self-stabilization.
\newblock {\em Theor. Comput. Sci.}, 186(1-2):199--229, 1997.

\bibitem[AW09]{AaronsonW09}
Scott Aaronson and Avi Wigderson.
\newblock Algebrization: {A} new barrier in complexity theory.
\newblock {\em {TOCT}}, 1(1):2:1--2:54, 2009.

\bibitem[BEG{\etalchar{+}}94]{BlumEGKN94}
Manuel Blum, William~S. Evans, Peter Gemmell, Sampath Kannan, and Moni Naor.
\newblock Checking the correctness of memories.
\newblock {\em Algorithmica}, 12(2/3):225--244, 1994.

\bibitem[BFLS91]{BabaiFLS91}
L{\'{a}}szl{\'{o}} Babai, Lance Fortnow, Leonid~A. Levin, and Mario Szegedy.
\newblock Checking computations in polylogarithmic time.
\newblock In {\em Proceedings of the 23rd Annual {ACM} Symposium on Theory of
  Computing, May 5-8, 1991, New Orleans, Louisiana, {USA}}, pages 21--31, 1991.

\bibitem[BFP15]{MorFP15}
Mor Baruch, Pierre Fraigniaud, and Boaz Patt{-}Shamir.
\newblock Randomized proof-labeling schemes.
\newblock In {\em Proceedings of the 2015 {ACM} Symposium on Principles of
  Distributed Computing, {PODC} 2015, Donostia-San Sebasti{\'{a}}n, Spain, July
  21 - 23, 2015}, pages 315--324, 2015.

\bibitem[BFS86]{BabaiFS86}
L{\'{a}}szl{\'{o}} Babai, Peter Frankl, and Janos Simon.
\newblock Complexity classes in communication complexity theory (preliminary
  version).
\newblock In {\em 27th Annual Symposium on Foundations of Computer Science,
  Toronto, Canada, 27-29 October 1986}, pages 337--347, 1986.

\bibitem[BG08]{BarakG08}
Boaz Barak and Oded Goldreich.
\newblock Universal arguments and their applications.
\newblock {\em {SIAM} J. Comput.}, 38(5):1661--1694, 2008.

\bibitem[BM88]{BabaiM88}
L{\'{a}}szl{\'{o}} Babai and Shlomo Moran.
\newblock Arthur-merlin games: {A} randomized proof system, and a hierarchy of
  complexity classes.
\newblock {\em J. Comput. Syst. Sci.}, 36(2):254--276, 1988.

\bibitem[CCH{\etalchar{+}}18]{CanettiCHLRR18}
Ran Canetti, Yilei Chen, Justin Holmgren, Alex Lombardi, Guy~N. Rothblum, and
  Ron~D. Rothblum.
\newblock Fiat-shamir from simpler assumptions.
\newblock Cryptology ePrint Archive, Report 2018/1004, 2018.
\newblock \url{http://eprint.iacr.org/}.

\bibitem[Con92]{Condon92}
Anne Condon.
\newblock The complexity of space boundes interactive proof systems.
\newblock In Klaus Ambos{-}Spies, Steven Homer, and Uwe Sch{\"{o}}ning,
  editors, {\em Complexity Theory: Current Research, Dagstuhl Workshop,
  February 2-8, 1992}, pages 147--189. Cambridge University Press, 1992.

\bibitem[DS92]{DworkS92}
Cynthia Dwork and Larry~J. Stockmeyer.
\newblock Finite state verifiers {I:} the power of interaction.
\newblock {\em J. {ACM}}, 39(4):800--828, 1992.

\bibitem[FFH16]{FeuilloleyFH16}
Laurent Feuilloley, Pierre Fraigniaud, and Juho Hirvonen.
\newblock A hierarchy of local decision.
\newblock In {\em 43rd International Colloquium on Automata, Languages, and
  Programming, {ICALP} 2016, July 11-15, 2016, Rome, Italy}, pages
  118:1--118:15, 2016.

\bibitem[FFH{\etalchar{+}}18]{FeuilloleyFHPP18}
Laurent Feuilloley, Pierre Fraigniaud, Juho Hirvonen, Ami Paz, and Mor Perry.
\newblock Redundancy in distributed proofs.
\newblock In {\em 32nd International Symposium on Distributed Computing, {DISC}
  2018, New Orleans, LA, USA, October 15-19, 2018}, pages 24:1--24:18, 2018.

\bibitem[FHK12]{FraigniaudHK12}
Pierre Fraigniaud, Magn{\'{u}}s~M. Halld{\'{o}}rsson, and Amos Korman.
\newblock On the impact of identifiers on local decision.
\newblock In {\em Principles of Distributed Systems, 16th International
  Conference, {OPODIS} 2012, Rome, Italy, December 18-20, 2012. Proceedings},
  pages 224--238, 2012.

\bibitem[FKP11]{fraigniaud2011local}
Pierre Fraigniaud, Amos Korman, and David Peleg.
\newblock Local distributed decision.
\newblock In {\em Foundations of Computer Science (FOCS), 2011 IEEE 52nd Annual
  Symposium on}, pages 708--717. IEEE, 2011.

\bibitem[FS86]{FiatS86}
Amos Fiat and Adi Shamir.
\newblock How to prove yourself: Practical solutions to identification and
  signature problems.
\newblock In {\em CRYPTO}, volume 263, pages 186--194. Springer, 1986.

\bibitem[GKR15]{GoldwasserKR15}
Shafi Goldwasser, Yael~Tauman Kalai, and Guy~N. Rothblum.
\newblock Delegating computation: Interactive proofs for muggles.
\newblock {\em J. {ACM}}, 62(4):27:1--27:64, 2015.

\bibitem[GMR89]{GoldwasserMR89}
Shafi Goldwasser, Silvio Micali, and Charles Rackoff.
\newblock The knowledge complexity of interactive proof systems.
\newblock {\em {SIAM} J. Comput.}, 18(1):186--208, 1989.

\bibitem[GMW91]{GoldreichMW91}
Oded Goldreich, Silvio Micali, and Avi Wigderson.
\newblock Proofs that yield nothing but their validity for all languages in
  {NP} have zero-knowledge proof systems.
\newblock {\em J. {ACM}}, 38(3):691--729, 1991.

\bibitem[GPW18]{GoosPW18}
Mika G{\"{o}}{\"{o}}s, Toniann Pitassi, and Thomas Watson.
\newblock The landscape of communication complexity classes.
\newblock {\em Computational Complexity}, 27(2):245--304, 2018.

\bibitem[GR18]{GurR18}
Tom Gur and Ron~D. Rothblum.
\newblock Non-interactive proofs of proximity.
\newblock {\em Computational Complexity}, 27(1):99--207, 2018.

\bibitem[GS89]{GoldwasserS89}
Shafi Goldwasser and Michael Sipser.
\newblock Private coins versus public coins in interactive proof systems.
\newblock {\em Advances in Computing Research}, 5:73--90, 1989.

\bibitem[GS16]{GoosS16}
Mika G{\"{o}}{\"{o}}s and Jukka Suomela.
\newblock Locally checkable proofs in distributed computing.
\newblock {\em Theory of Computing}, 12(1):1--33, 2016.

\bibitem[HT74]{HopcroftT74}
John~E. Hopcroft and Robert~Endre Tarjan.
\newblock Efficient planarity testing.
\newblock {\em J. {ACM}}, 21(4):549--568, 1974.

\bibitem[IKOS08]{IshaiKOS08}
Yuval Ishai, Eyal Kushilevitz, Rafail Ostrovsky, and Amit Sahai.
\newblock Cryptography with constant computational overhead.
\newblock In {\em Proceedings of the 40th Annual {ACM} Symposium on Theory of
  Computing, Victoria, British Columbia, Canada, May 17-20, 2008}, pages
  433--442, 2008.

\bibitem[Kil92]{Kilian92}
Joe Kilian.
\newblock A note on efficient zero-knowledge proofs and arguments (extended
  abstract).
\newblock In {\em STOC}, pages 723--732. {ACM}, 1992.

\bibitem[KK07]{KormanK07}
Amos Korman and Shay Kutten.
\newblock Distributed verification of minimum spanning trees.
\newblock {\em Distributed Computing}, 20(4):253--266, 2007.

\bibitem[KKP10]{KormanKP10}
Amos Korman, Shay Kutten, and David Peleg.
\newblock Proof labeling schemes.
\newblock {\em Distributed Computing}, 22(4):215--233, 2010.

\bibitem[KOS18]{KolOS18}
Gillat Kol, Rotem Oshman, and Raghuvansh~R. Saxena.
\newblock Interactive distributed proofs.
\newblock In {\em Proceedings of the 2018 {ACM} Symposium on Principles of
  Distributed Computing, {PODC} 2018, Egham, United Kingdom, July 23-27, 2018},
  pages 255--264, 2018.

\bibitem[KPY18]{KalaiPY18}
Yael Kalai, Omer Paneth, and Lisa Yang.
\newblock On publicly verifiable delegation from standard assumptions.
\newblock {\em {IACR} Cryptology ePrint Archive}, 2018.

\bibitem[KR08]{KalaiR08}
Yael~Tauman Kalai and Ran Raz.
\newblock Interactive {PCP}.
\newblock In {\em Automata, Languages and Programming, 35th International
  Colloquium, {ICALP} 2008, Reykjavik, Iceland, July 7-11, 2008, Proceedings,
  Part {II} - Track {B:} Logic, Semantics, and Theory of Programming {\&} Track
  {C:} Security and Cryptography Foundations}, pages 536--547, 2008.

\bibitem[KW90]{KarchmerW90}
Mauricio Karchmer and Avi Wigderson.
\newblock Monotone circuits for connectivity require super-logarithmic depth.
\newblock {\em {SIAM} J. Discrete Math.}, 3(2):255--265, 1990.

\bibitem[LFKN92]{LundFKN92}
Carsten Lund, Lance Fortnow, Howard~J. Karloff, and Noam Nisan.
\newblock Algebraic methods for interactive proof systems.
\newblock {\em J. {ACM}}, 39(4):859--868, 1992.

\bibitem[Mic00]{Micali00}
Silvio Micali.
\newblock Computationally sound proofs.
\newblock {\em {SIAM} J. Comput.}, 30(4):1253--1298, 2000.

\bibitem[Osh18]{Oshman18}
Rotem Oshman.
\newblock Distributed interactive proofs, 2018.
\newblock Interactive Complexity Workshop, Simon's Institute.
  \url{https://simons.berkeley.edu/talks/distributed-interactive-proofs}.

\bibitem[RRR16]{ReingoldRR16}
Omer Reingold, Guy~N. Rothblum, and Ron~D. Rothblum.
\newblock Constant-round interactive proofs for delegating computation.
\newblock In {\em Proceedings of the 48th Annual {ACM} {SIGACT} Symposium on
  Theory of Computing, {STOC} 2016, Cambridge, MA, USA, June 18-21, 2016},
  pages 49--62, 2016.

\bibitem[RVW13]{RothblumVW13}
Guy~N. Rothblum, Salil~P. Vadhan, and Avi Wigderson.
\newblock Interactive proofs of proximity: delegating computation in sublinear
  time.
\newblock In Dan Boneh, Tim Roughgarden, and Joan Feigenbaum, editors, {\em
  Symposium on Theory of Computing Conference, STOC'13, Palo Alto, CA, USA,
  June 1-4, 2013}, pages 793--802. {ACM}, 2013.

\bibitem[Sha92]{Shamir92}
Adi Shamir.
\newblock {IP} = {PSPACE}.
\newblock {\em J. {ACM}}, 39(4):869--877, 1992.

\bibitem[Wil16]{Williams16}
Richard~Ryan Williams.
\newblock Strong {ETH} breaks with merlin and arthur: Short non-interactive
  proofs of batch evaluation.
\newblock In {\em 31st Conference on Computational Complexity, {CCC} 2016, May
  29 to June 1, 2016, Tokyo, Japan}, pages 2:1--2:17, 2016.

\end{thebibliography}


\end{document}